\ifx\synctex\undefined\else\synctex=1\fi
\pdfoutput=1
\documentclass{llncs}

\usepackage{color}
\usepackage{graphicx}
\usepackage{enumerate}
\usepackage[fleqn]{amsmath} 
\usepackage{algorithm}
\usepackage[noend]{algpseudocode}
\usepackage{alltt}
\usepackage{tikz}
\usepackage{wrapfig}
\usetikzlibrary{arrows,snakes,backgrounds, shapes, positioning,calc,intersections}
\newcommand\irregularcircle[2]{% radius, irregularity  (circle)
  let \n1 = {(#1)+rand*(#2)} in
  +(0:\n1)
  \foreach \a in {10,20,...,350}{
    let \n1 = {(#1)+rand*(#2)} in
    -- +(\a:\n1)
  } -- cycle
}
\usepackage{adjustbox} 
\usepackage{multicol}
\usepackage{amssymb}
\usepackage{upquote}
\usepackage{fancyvrb} %for BVerbatim environment
\usepackage[belowskip=-15pt,aboveskip=5pt]{caption} %remove space after figure caption
\usepackage{float}
 \usepackage{wrapfig}

\usepackage{hyperref}

\usepackage{upquote} %for underlining within listings
\usepackage{listings}% http://ctan.org/pkg/listings
\usepackage[final,bfstyle]{commenting}
\declareauthor{bk}{Bish}{red!70!black}
\authorcommand{bk}{comment}
\declareauthor{pds}{Scha}{green!80!black}
\authorcommand{pds}{comment}
\declareauthor{pjs}{Stu}{brown!80!black}
\authorcommand{pjs}{comment}
\declareauthor{hs}{Harald}{purple!80!black}
\authorcommand{hs}{comment}
\declareauthor{gg}{Graeme}{orange!80!black}
\authorcommand{gg}{comment}

\newcommand{\ignore}[1]{}

\newcommand{\pcode}[2][\codesize]{
    \fbox{
    \begin{minipage}{0.45\linewidth}
    #1
    \begin{tabbing}
    xx \= xx \= xx \= xx \= xx \= xx \= xx \= xx \= xx \= xx \= \kill
    #2
    \end{tabbing}
    \end{minipage}
    }
  }

\title{Precondition Inference via \\ Partitioning of Initial States}

\author %[Kafle, Gange, Schachte, S{\o}ndergaard, Stuckey]
{Bishoksan Kafle \inst{1}
\and
Graeme Gange \inst{2}
\and 
Peter Schachte  \inst{1} 
\and \\
Harald S{\o}ndergaard  \inst{1}
\and
Peter J. Stuckey  \inst{2}
}
\institute{The University of Melbourne
\and
 Monash university}

\begin{document}
\maketitle

\newcommand{\rahft}{\textsc{Rahft}}
\newcommand{\wprahft}{\textsc{WP-Rahft}}
\newcommand{\pihorn}{\textsc{PI-Horn}}
\newcommand{\integ}{{\sf int}}
\newcommand{\listint}{{\sf listint}}
\newcommand{\other}{{\sf other}}
\newcommand{\true}{\textsc{True}}
\newcommand{\false}{\textsc{False}}
\newcommand{\trueit}{\mathit{true}}
\newcommand{\falseit}{\mathit{false}}
\newcommand{\falsett}{\mathtt{false}}
\newcommand{\Bin}{{\sf Bin}}
\newcommand{\Dep}{{\sf Dep}}
\newcommand{\g}{{\sf g}}
\newcommand{\nong}{{\sf ng}}
\newcommand{\OL}{{\cal O}}
\newcommand{\M}{{\sf M}}
\newcommand{\R}{{\cal R}}
\newcommand{\A}{\mathcal{A}}

\newcommand{\body}{\mathcal{B}}
\newcommand{\B}{{\cal B}}
\newcommand{\C}{{\cal C}}
\newcommand{\D}{{\cal D}}
\newcommand{\X}{{\cal X}}
\newcommand{\V}{{\cal V}}
\newcommand{\Q}{{\cal Q}}
\newcommand{\F}{{\sf F}}
\newcommand{\N}{{\cal N}}
\newcommand{\Lang}{{\cal L}}
\newcommand{\powerset}{{\cal P}}
\newcommand{\FTA}{{\cal FT\!A}}
\newcommand{\Term}{{\sf Term}}
\newcommand{\Empty}{{\sf empty}}
\newcommand{\nonEmpty}{{\sf nonempty}}
\newcommand{\compl}{{\sf complement}}
\newcommand{\args}{{\sf args}}
\newcommand{\preds}{{\sf preds}}
\newcommand{\gnd}{{\sf gnd}}
\newcommand{\lfp}{{\sf lfp}}
\newcommand{\psharp}{P^{\sharp}}
\newcommand{\minimize}{{\sf minimize}}
\newcommand{\headterms}{\mathsf{headterms}}
\newcommand{\solvebody}{\mathsf{solvebody}}
\newcommand{\solve}{\mathsf{solve}}
\newcommand{\fail}{\mathsf{fail}}
\newcommand{\member}{\mathsf{memb}}
\newcommand{\ground}{\mathsf{ground}}
\newcommand{\abst}{\mathsf{abstract}}
\newcommand{\bodyfacts}{\mathsf{bodyfacts}}
\newcommand{\renameunfold}{\mathsf{renameunfold}}
\newcommand{\rep}{\mathsf{rep}}
\newcommand{\cfacts}{\mathsf{cfacts}}

\newcommand{\raf}{{\sf raf}}
\newcommand{\qa}{{\sf qa}}
\newcommand{\spl}{{\sf split}}

%FOR TRANSFORMATIONS
\newcommand{\tr}{{\sf tr}}
\newcommand{\trA}[1][A]{\ensuremath{\mathsf{tr}_{#1}}}
\newcommand{\ris}{{\sf u\_approximate}}
\newcommand{\replace}{{\sf init\_replace}}

\newcommand{\transitions}{\mathsf{transitions}}
\newcommand{\nonempty}{\mathsf{nonempty}}
\newcommand{\dom}{\mathsf{dom}}

\newcommand{\Args}{\mathsf{Args}}
\newcommand{\id}[1]{\mathit{#1}}
\newcommand{\type}{\tau}
\newcommand{\restrict}{\mathsf{restrict}}
\newcommand{\any}{\top}
\newcommand{\dyn}{\top}
\newcommand{\dettypes}{{\sf dettypes}}
\newcommand{\Atom}{{\sf Atom}}

\newcommand{\chc}{{\sf chc}}
\newcommand{\deriv}{{\sf deriv}}

\newcommand{\vars}{\mathsf{vars}}
\newcommand{\Vars}{\mathsf{Vars}}
\newcommand{\range}{\mathsf{range}}
\newcommand{\varpos}{\mathsf{varpos}}
\newcommand{\varid}{\mathsf{varid}}
\newcommand{\argpos}{\mathsf{argpos}}
\newcommand{\elim}{\mathsf{elim}}
\newcommand{\pred}{\mathsf{pred}}
\newcommand{\predfuncs}{\mathsf{predfuncs}}
\newcommand{\project}{\mathsf{project}}
\newcommand{\reduce}{\mathsf{reduce}}
\newcommand{\positions}{\mathsf{positions}}
\newcommand{\contained}{\preceq}
\newcommand{\equivalent}{\cong}
\newcommand{\unify}{{\it unify}}
\newcommand{\Iff}{{\rm iff}}
\newcommand{\Where}{{\rm where}}
\newcommand{\qmap}{{\sf qmap}}
\newcommand{\fmap}{{\sf fmap}}
\newcommand{\ftable}{{\sf ftable}}
\newcommand{\Qmap}{{\sf Qmap}}
\newcommand{\states}{{\sf states}}
\newcommand{\head}{\tau}
\newcommand{\atomconstraints}{\mathsf{atomconstraints}}
\newcommand{\thresholds}{\mathsf{thresholds}}
\newcommand{\term}{\mathsf{Term}}
\newcommand{\trees}{\mathsf{trees}}
\newcommand{\renames}{\rho}
\newcommand{\renameps}{\rho_2}
\newcommand{\predicates}{\mathsf{Predicates}}
\newcommand{\query}{\mathsf{q}}
\newcommand{\ans}{\mathsf{a}}
\newcommand{\trace}{\mathsf{tr}}
\newcommand{\constr}{\mathsf{constr}}
\newcommand{\Iproj}{\mathsf{proj}}
\newcommand{\SAT}{\mathsf{SAT}}
\newcommand{\interpolant}{\mathsf{interpolant}}
\newcommand{\unknown}{?}
\newcommand{\rhs}{{\sf rhs}}
\newcommand{\lhs}{{\sf lhs}}
\newcommand{\unfold}{{\sf unfold}}
\newcommand{\arity}{{\sf ar}}
\newcommand{\AND}{{\sf AND}}

\newcommand{\feasible}{{\sf feasible}}
\newcommand{\infeasible}{{\sf infeasible}}
\newcommand{\safe}{{\sf safe}}
\newcommand{\unsafe}{{\sf unsafe}}
\newcommand{\init}{{\sf init}}
\newcommand{\theory}{\mathbb{T}}
\newcommand{\trseq}{{\textsf{tr-seq}}}
\newcommand{\trseqA}[1][A]{\ensuremath{\textsf{tr-seq}_{#1}}}

\renewcommand{\phi}{\varphi}
\newcommand{\suffpre}[1]{\phi^{S}_{#1}}
\newcommand{\necpre}[1]{\phi^{N}_{#1}}

\newcommand{\swp}{{\suffpre{s}}} %safe weakest precondition
\newcommand{\spec}{{\sf spec}} %specialisation

\newcommand{\atmost}[1]{\le #1}
\newcommand{\exactly}[1]{=#1}
\newcommand{\exceeds}[1]{>#1}
\newcommand{\anydim}[1]{\ge 0}

\def\ll{[\![}
\def\rr{]\!]}

\newcommand{\sset}[2]{\left\{~#1  \left|
                               \begin{array}{l}#2\end{array}
                          \right.     \right\}}

\newcommand{\qin}{\hspace*{0.15in}}
\newenvironment{SProg}
     {\begin{small}\begin{tt}\begin{tabular}[t]{l}}%
     {\end{tabular}\end{tt}\end{small}}
\def\anno#1{{\ooalign{\hfil\raise.07ex\hbox{\small{\rm #1}}\hfil%
        \crcr\mathhexbox20D}}}

%\newtheorem{definition}{Definition}
%\newtheorem{example}{Example}
%\newtheorem{corollary}{Corollary}
%
%\newtheorem{lemma}{Lemma}
%\newtheorem{theorem}{Theorem}
%\newtheorem{proposition}{Proposition}
%\newtheorem{property}{Property}

%tuple of args in angle brackets
\newcommand{\tuple}[1]{\langle #1 \rangle} 
% use in math mode to indicate a metasyntactic variable standing for a tuple
\newcommand{\tuplevar}[1]{\mathbf{#1}}

\begin{abstract}
 Precondition inference is a non-trivial task with several applications in program analysis and verification. We present a novel iterative method for automatically deriving sufficient
preconditions for safety and unsafety of programs which introduces a new 
dimension of modularity. Each iteration maintains over-approximations of the set of \emph{safe} and \emph{unsafe} \emph{initial} states. 
%Existing methods refine these abstractions  while  the program remains the same. 
Then we repeatedly use the current abstractions to partition the program's \emph{initial} states into those known to be safe, known to be unsafe and unknown, and construct a revised program focusing on those initial states that are not yet known  to be safe or unsafe. An experimental evaluation of the method on a set of software verification benchmarks shows that it can solve problems which are not solvable using  previous methods.

\end{abstract}

\bk{
- not self-contained (because the way examples and transformations
in Section 3 are presented do not allow a reader to understand how they work)
-transformations are explained insufficiently
-verification of technical soundness is complicated as the proofs (or their
sketches) are not given by the authors. Both issues, however, can be solved if
authors will publish a technical report with those details
}

\section{Introduction}\label{sec:intro}
Precondition analysis infers conditions (on the initial states of a program)  that establish runtime properties of interest. For example, a \emph{sufficient} precondition for safety is a set of initial states, each of which is guaranteed to be safe with respect to given safety properties (assertions). 
Preconditions have several important applications and shed  light on the following questions: (i) which inputs ensure safe program execution? (program analysis); (ii) which inputs can or cannot cause an error? (program debugging); (iii) what are the valid inputs for a program? (program understanding); and  (iv) which library functions need to be modified to ensure valid outputs? (program specialisation).

However, the problem is undecidable in general, hence approaches to finding preconditions use approximations. Preconditions derived by over-approximations are usually too weak, while under-approximations are usually too strong to be useful in practice.  Unfortunately, for reasons of computability, we cannot hope to obtain an exact (optimal)  precondition (though such preconditions enable their re-use, for example, under different calling contexts), so we must instead look for a (weaker) precondition that is useful. 
This work builds upon the transformation-guided framework of Kafle et al.~\cite{kafle-iclp18}.  Additionally, we  model the set of safe terminating states so that we can iteratively refine  the approximations of the set of safe states as well (as in ~\cite{kafle-iclp18}). We also   partition  initial states guided by the over-approximations of the safe and unsafe states to derive  more precise sufficient preconditions, and in some cases precise necessary and sufficient preconditions (optimal). 
%Furthermore, our iterative approach can flexibly be stopped any time, to produce a sound precondition. 

We use constrained Horn clauses (CHCs) to encode the computation of C programs.
CHCs can conveniently be used to represent big- and small-step semantics, transition systems, imperative programs and other models of computation. The use of CHCs allows decoupling programming language syntax and semantics from the analysis. We use two special predicates, $\safe$ and $\unsafe$ respectively to encode the set of safe and the set of unsafe (error) states. 
A safe state is reached in all paths without error (return statements are encoded by $\safe$ predicate). The predicate \init~ encodes the set of initial states. We assume that the users specify all the states of interest by an appropriate construct provided by the language (for example, \emph{assert(c), return $\tuple{n}$} of the C language). States not specified by the user (for example, \emph{buffer-overflow, floating point exceptions}) are not taken into account while generating CHCs. Hence the correctness of the preconditions depends on the user specified set of states. 

\noindent
\textbf{Running example.} Before formally presenting our approach, we illustrate the main steps via the example program in Figure~\ref{ex:precond}.
The left side shows a C program fragment, and the right its CHC
representation, encoding the reachable states of the computation. 
Program variables in C are represented by logical variables (Prolog style capital letters) in CHC. The clause $c_1$ specifies the initial states of the program via the predicate \emph{init} which is always reachable. Similarly, the clauses $c_2$ and $c_3$ encode the reachability of the \emph{while} loop via the predicate \emph{wh}. $c_2$ states that the loop is reachable if \emph{init} is reachable while $c_3$ states that the loop is reachable if it is already in the loop and the loop guard is satisfied (recursive case). 
The first three clauses on the right define the behaviour of the program in 
Figure~\ref{ex:precond} (left) and the last two clauses represent the  
properties of the program. The clause $c_4$ states that an ``unsafe'' state is reached if $\mathtt{B<0}$ after exiting the loop whereas the clause $c_5$ states that the program terminates gracefully if $\mathtt{B \geq 0}$. Note that the semantics of \emph{assert(c)} is \emph{if(c)} \texttt{SKIP} \emph{else} \texttt{ERROR}.
CHCs can be obtained from an imperative program using various approaches 
\cite{Peralta-Gallagher-Saglam-SAS98,DBLP:conf/pldi/GrebenshchikovLPR12,DBLP:conf/cav/GurfinkelKKN15,DBLP:journals/scp/AngelisFPP17}.  Henceforth whenever we refer to a program, we refer to its CHC representation. 

%\bk{Why CHC encoding in Fig. 1 declares that "unsafe" is reached if $B<0$ after the
%execution of the loop, even if the assertion in the original does not require
%execution of the loop? In case the loop was executed properly and terminated
%before $a$ got less than 1, e.g. due to $if$s combined with  $break$ and/or
%$goto$ statements in the loop. Then, "assert( $b \geq 0$ )" will fire only when
%its condition is satisfied, regardless of the value of $a$. Therefore, the
%integrity constraints in the encoding require more than the original assertion
%in the given program. Of course, given the simplicity of the example program, it
%is obvious that the generated CHC encoding is correct, but it is unclear how it
%will work in a general case. This specific handling of the loop condition makes
%the example hard to understand.
%-- To simplify the introduction please add an explanation of why $A\leq 0$ is
%dropped, when the disjunction of SPs over the iterations is generated.}
\begin{figure}[t]
  \begin{center}
  \begin{tabular}{ll}
    \pcode[\small]{
   %$\textbf{int}~ a=0, \textbf{int}~b=0$; \\
   \textbf{int} main(\textbf{int}~ a, \textbf{int}~ b) \{ \\
   ~~\textbf{while} ($a \geq 1 $) \{  \\
   ~~~~$a=a-1;$ 
      	 $b=b-1;$ \\    
   ~~\}~~ \\% \textbf{assert} ($b \geq 0$); \\
   ~~\textbf{assert} ($b \geq 0$); \\
   \} 
      }    &
    \pcode[\small]{
$\mathtt{c_1.~ init(A,B)}.$ \\ %\% ~A=0, B=0.}$ \\ ~~~
$\mathtt{c_2.~ wh(A,B) \leftarrow init(A,B).}$\\
$\mathtt{c_3.~ wh(A,B) \leftarrow A_0\geq1, A=A_0-1,}$ \\
$\mathtt{~~~~~~~~~~~~~~~~~~~~ B=B_0-1, wh(A_0,B_0).}$\\
$\mathtt{c_4.~ unsafe\leftarrow A<1, B<0, wh(A,B).}$ \\
$\mathtt{c_5.~ safe\leftarrow A<1, B\geq0, wh(A,B).}$
    } 
  \end{tabular}
  \end{center}
  \caption{Running example: (left) original program, (right) translation to
   CHCs}
 \label{ex:precond}
\end{figure}
Observe that the  assertion is violated if the initial conditions on $a$ and $b$ entail $\mathtt{(b<0 \wedge a \leq 0) \vee (a\geq 1 \wedge a> b)}$, and the program terminates gracefully if  $\mathtt{(b\ge0 \wedge a\leq 0) \vee (a\geq 1 \wedge b\geq a)}$. 
Automatic derivation of these preconditions is challenging for the following reasons:
 (i)   the desired result is a disjunction of linear constraints---so we need a domain that can express disjunctions;
(ii) information has to be propagated forwards and backwards because we 
      lose information on $b$ and $a$ in the forward and in the backward direction respectively;
(iii) we need to reason simultaneously about the safe and unsafe states; one cannot simply be obtained by complementing the other. Existing methods \cite{DBLP:conf/lopstr/HoweKL04,DBLP:conf/sas/BakhirkinBP14,DBLP:conf/vmcai/Moy08,DBLP:journals/entcs/Mine12,kafle-iclp18} cannot infer the desired preconditions.

Let $\necpre{s}$ and $\necpre{u}$ represent the over-approximations of safe and unsafe initial states, that is, necessary preconditions (NPs) for safety and unsafety, respectively. From the program in Figure~\ref{ex:precond}, we get $\necpre{s}=\necpre{u}=\trueit$ (the set of initial states). Thus the sufficient preconditions (SPs) for: (i) safety $\suffpre{s}= \necpre{s} \wedge \neg \necpre{u}$ and (ii) unsafety $\suffpre{u}= \necpre{u} \wedge \neg \necpre{s}$ are both $\falseit$.  These SPs, although valid, are uninteresting. So, we transform the program with respect to the predicates \safe\ and \unsafe\ using CHC transformations described in Section \ref{sec:specialisation} (with the aim of strengthening NPs) and 
derive $\necpre{s} \equiv \mathtt{(B\ge0 \wedge A\leq 0) \vee (A\geq 1 \wedge B\geq 0)}\equiv \mathtt{B \geq 0}$ and $\necpre{u} \equiv \mathtt{(B<0 \wedge A\leq 0) \vee A\geq 1}\equiv \mathtt{B \leq 0 \vee A \geq 1}$ from the resulting programs. The SPs from this iteration are 
$\suffpre{s} \equiv \mathtt{B \geq 0 \wedge A \leq 0}$ and
$\suffpre{u} \equiv \mathtt{B <0}$. 
These NPs are more precise than the one obtained in the previous step. However further refinement is possible since the NPs for safety and unsafety overlap, that is, $\necpre{s} \wedge \necpre{u}= \mathtt{B\geq 0 \wedge A \ge 1}$ is \emph{satisfiable}. Since the formulae $\necpre{s} \wedge \neg \necpre{u}$ and $\necpre{u} \wedge \neg \necpre{s}$ are already proven to be  \emph{safe} and \emph{unsafe} initial states respectively, we can focus our attention on $\necpre{s} \wedge \necpre{u}$, which is yet to be classified as \emph{safe} or \emph{unsafe}. To perform further classification, we constrain the initial states of the program in Figure \ref{ex:precond} to $\mathtt{B\geq 0 \wedge A \ge 1}$ and restart the analysis by specialising the programs with respect to \safe~and \unsafe.  As a result, 
we derive $\necpre{s} \equiv \mathtt{(B\ge A \wedge A\geq 1)}$ and $\necpre{u} \equiv \mathtt{(B \ge 0 \wedge A > B)}$ as NPs.
Since $\necpre{s} \wedge \necpre{u} \equiv \falseit$, the preconditions are optimal (they are both sufficient and necessary for the validity of assertions) and the algorithm terminates. The final SPs are given (as the disjunction of SPs over the iterations) as follows (after simplifications):
\begin{flalign*}
\suffpre{s} \equiv \mathtt{(B\ge 0 \wedge A \leq 0) \vee (B\ge A \wedge A\geq 1).}  ~~~~
\suffpre{u} \equiv \mathtt{B<0 \vee (B \ge 0 \wedge A > B).}
\end{flalign*}
The key contributions of the paper  are as follows. 

\begin{itemize}
\item We extend the transformation-guided framework of Kafle et al.~\cite{kafle-iclp18} which can now be used to derive SPs not only for safety but also for unsafety using over-approximation techniques. This is enabled by the encoding of the \emph{safe} set of states together with the \emph{unsafe} set of states and their simultaneous over-approximations. It also allows detecting optimality as well  as controlling precision of the derived preconditions. In addition, it  allows us to derive precondition for  non-termination. 
\item We partition the initial states of the program (guided by over-approximations) into those known to be safe, known to be unsafe, and whose safety are unknown. This enables us to reason incrementally on the new partition that is yet to be proven safe or unsafe. The precondition for the original program can then be derived by composing preconditions of each partitions. 
%\item The framework allows us to reason about  non-termination which is enabled by the simultaneous reasoning of the \emph{safe} and \emph{unsafe} sets of states.  
\end{itemize}  
%a framework for deriving SPs for safety and unsafety whose precision can be controlled, and which has the ability to detect optimality. Firstly, we use over-approximations of both the safe and unsafe sets of states and refine them using off-the-shelf tools and techniques from the literature. Secondly, we partition the initial states of the program into those known to be safe, known to be unsafe, and whose safety are unknown. This enables us to reason incrementally on the new partition that is yet to be proven safe or unsafe. The framework achieves this simplicity by avoiding the need to compute weakest preconditions or rely on abstract domains with special properties or intricate transfer functions as in previous works \cite{DBLP:conf/lopstr/HoweKL04,DBLP:conf/sas/BakhirkinBP14,DBLP:conf/vmcai/Moy08,DBLP:journals/entcs/Mine12}.  

\paragraph{Paper outline.} 
After preliminaries in Section~\ref{sec:prelim}, we describe preconditions, 
CHC transformations and their relationships in 
Section~\ref{sec:specialisation}. 
Section \ref{sec:precond} presents our algorithm for precondition inference, 
and Section~\ref{sec:experiments} is an account of experimental evaluation.
Section~\ref{sec:rel} places the work in context and Section~\ref{sec:conclusion} concludes.

\section{Preliminaries}
\label{sec:prelim}
An atomic formula (simply \emph{atom}), is a formula $p(\tuplevar{x})$ 
where $p$ is a predicate symbol and $\tuplevar{x}$ a tuple of arguments. 
A constrained Horn clause (CHC) is a first-order formula written as $p_0(\tuplevar{x}_0) \leftarrow \phi, 
p_1(\tuplevar{x}_1), \ldots, p_k(\tuplevar{x}_k)$ following Constraint Logic Programming (CLP) standard, 
where $\phi$ is a finite conjunction of quantifier-free \emph{constraints} 
on variables $\tuplevar{x}_i$ with respect to some constraint theory 
$\theory$, $p_i(\tuplevar{x}_i)$ are atoms and  $p_0(\tuplevar{x}_0)$ and $\phi, 
p_1(\tuplevar{x}_1), \ldots, p_k(\tuplevar{x}_k)$ are respectively called the \emph{head} and the \emph{body} of the clause.

In textual form,  we use Prolog-like syntax and typewriter font, 
with  capital letters for variable names and linear arithmetic constraints. A \emph{constrained fact} is a clause of the form  $p_0(\tuplevar{x}_0) \leftarrow \phi$, where $\phi \in \theory$.  A clause of the form $\unsafe \leftarrow \ldots$ is called an \emph{integrity constraint} where the predicate $\unsafe$ is always interpreted as $\falseit$. 
A set of CHCs is usually called a program. 
  
\paragraph{CHC semantics.} 

The semantics of CHCs is derived  from the semantics of predicate logic.  
An \emph{interpretation} assigns to each predicate a relation over the
domain of the constraint theory $\theory$. An interpretation that satisfies  each clause in the set is called a model of a set of CHCs. Note that the set of CHCs without any \emph{integrity constraint}  always has a model (an interpretation that assigns \emph{true} to each atom on the head of a clause).  When modelling safety properties of systems using CHCs, we also refer to CHCs as being \emph{safe} or \emph{unsafe} when 
they have a model, or do not have a model, respectively. 
We assume that the theory $\theory$ is equipped with a decision procedure and a  projection operator, and that it is closed under negation.
We use notation $\phi|_{V}$ to represent the constraint formula $\phi$
projected onto variables $V$. 
 \begin{definition}[AND-tree or derivation tree, adapted from \cite{Gallagher-Lafave-Dagstuhl}]\label{def:andtree}
An \emph{AND-tree} for a set of CHCs is a labelled tree whose nodes are labelled as follows.
\begin{enumerate}
\item
each non-leaf node corresponds to a clause (with variables suitably renamed) of the form
$A \leftarrow \phi, A_1,\ldots,A_k$ and is labelled by an atom $A, \phi$, and has children labelled by $A_1,\ldots,A_k$,
\item
each leaf node corresponds to a clause of the form $A \leftarrow \phi$ (with variables suitably renamed) and is labelled by
an atom $A$ and $\phi$, and
%\item each node is labelled with the clause identifier of the corresponding clause.
\end{enumerate}
 \end{definition}
%CHC derivation trees are represented by
%AND-trees; we refer to Gallagher and Lafave \cite{Gallagher-Lafave-Dagstuhl} for formal definition. Henceforth we use  AND-trees and derivation trees interchangably.  
Given a derivation tree $t$, $\constr(t)$ represents the conjunction of its constraints. The tree $t$ is \emph{feasible} if and only if $\constr(t)$ is satisfiable  over $\theory$.

\begin{definition}[Initial clauses and nodes]\label{def:initial-node}
Let $P$ be a set of CHCs, with a
{distinguished} predicate 
$p^{I}$ in $P$ which we call the \emph{initial predicate}.
The \emph{constrained facts}
$\{(p^I(\tuplevar{x}) \leftarrow \theta) 
\mid (p^I(\tuplevar{x}) \leftarrow \theta) \in P ~\text{and}~ \theta \in \theory\}$ are called
the \emph{initial clauses} of $P$.
Let $t$ be an AND-tree for $P$.
A node labelled by 
$p^I(\tuplevar{x}) \leftarrow \theta$ is an \emph{initial node} of $t$.
We extend the term ``initial predicate" and use the symbol $p^I$ to 
refer also to renamed versions of the initial predicate that arise 
during clause transformations.
\end{definition}

\section{Preconditions and CHC transformations}
\label{sec:specialisation}

%We limit our attention to deriving a sufficient precondition 
%for safety of a program; the same for unsafety can be derived analogously. 
We limit our attention to the sets of clauses for which every derivation tree for 
 $\safe$ and $\unsafe$  (whether feasible or infeasible) has at least one initial node. 
%Although it is not decidable for an arbitrary set of CHCs $P$ whether every 
%derivation of  $\safe$ ($\unsafe$) uses the initial predicate, 
%the above condition on derivation trees
%can be checked syntactically from the predicate dependency graph for $P$. 
\begin{definition}[Transforming CHCs by replacing initial states]\label{def:repis}
Let $P$ be a set of CHCs and $\phi$ be a constraint over $\theory$. Let $P'$ be the set of 
clauses obtained from $P$ by replacing the initial clauses 
$\{(p^I(\tuplevar{x}) \leftarrow \theta_i) \mid 1 \le i \le k\}$ by 
$\{(p^I(\tuplevar{x}) \leftarrow  \phi) \}$.   This operation is denoted by $\replace(P,\phi)$.
\end{definition}

\begin{definition}[NP and SP for safety]\label{def:nec-suff-precond}
Let $P$ be a set of CHCs and  $\phi$ be a constraint over $\theory$. 
Then 
\begin{itemize}
\item $\phi$ is an NP for safety of $P$ $(\necpre{s}{(P)})$ if  
$P \vdash_\theory \safe$ entails $\replace(P,\phi) \vdash_\theory \safe$. 
In other words, $\phi$ (possibly $\trueit$) is an over-approximation of the set of 
initial states of $P$ that can reach \safe. 

\item $\phi$ is an SP for safety of $P$ $(\suffpre{s}{(P)})$ if $\replace(P,\phi) \not \vdash_\theory \unsafe$.  In other words, $\phi$ (possibly $\falseit$) represents an under-approximation of the set of initial states of $P$ that cannot reach \unsafe. 
\end{itemize}

%\begin{definition}[NP and SP for safety]\label{def:nec-suff-precond}
%Let $P$ be a set of CHCs.
%Let $\phi$ be a constraint over $\theory$, and let $P'$ be the set of 
%clauses obtained from $P$ by replacing the initial clauses 
%$\{(p^I(\tuplevar{x}) \leftarrow \theta_i) \mid 1 \le i \le k\}$ by 
%$\{(p^I(\tuplevar{x}) \leftarrow  \phi) \}$. 
%Then 
%\begin{itemize}
%\item $\phi$ is an NP for safety of $P$ $(\necpre{s}{(P)})$ if  
%$P \vdash_\theory \safe$ entails $P' \vdash_\theory \safe$. 
%In other words, $\phi$ is an over-approximation of the set of 
%initial states of $P$ that can reach \safe. 
%
%\item $\phi$ is an SP for safety of $P$ $(\suffpre{s}{(P)})$ if it is either the formula $\falseit$\ or satisfies
%$P' \vdash_\theory \safe$ and $P' \not \vdash_\theory \unsafe$.  In other words, $\phi$ represents an under-approximation of the set of initial states of $P$ that can reach \safe~but definitely not \unsafe\ or neither. 
%\end{itemize}

\end{definition}
Thus an SP for safety is a constraint that  
is sufficient to  block derivations of $\unsafe$ (given that we assume 
clauses for which $p^I$ is essential for any derivation of  $\unsafe$). However in practice, we would like to consider SP for safety as a constraint that  
is sufficient to allow derivations of \safe~and block derivations of $\unsafe$.
Analogously, we can define both  NP ($\necpre{u}{(P)}$) and SP ($\suffpre{u}{(P)}$)  for unsafety.
%Ideally, we would like to find the most general, or \emph{weakest}  preconditions.  
%It is not computable in general, so we aim to find a condition that is as weak as possible.
%The constraint $\falseit$ is always a safe sufficient precondition for safety and unsafety, albeit an uninteresting one. 
%On the other hand, if $ P \not\vdash_{\theory} \unsafe$ then any constraint, including $\trueit$, 
%is a safe precondition for $P$.  
In the following, we  show how an NP and an SP can be derived from a set of clauses.

\begin{lemma}[NP extracted from  clauses] 
\label{def:necsafe} 
Let $P$ be a set of clauses encoding reachable states of a program. 
Then the  formula 
$
\necpre{s}(P)= \bigvee \{\theta \mid  (p^I(\tuplevar{x}) \leftarrow \theta) \in P\} 
$ is an NP for safety. 
\end{lemma}
\begin{proof}[sketch]
 Let $P'= \replace(P,\necpre{s}(P))$. 
 To show: $P \vdash_\theory \safe$ entails $P' \vdash_\theory \safe$.

\noindent
Let $P \vdash_\theory \safe$. Then there exists a derivation tree $t$ rooted at \safe~such that $\constr(t)$ is satisfiable in $\theory$. From our assumption, $t$ uses some initial nodes at its leaves. Since all  constraints from initial nodes of $P$ are in $\necpre{s}(P)$ and are used in constructing initial clauses of $P'$, $t$ is a feasible derivation of   $P'$ as well. Then we have  
 $P' \vdash_\theory \safe$. \qed
\end{proof}
An NP for unsafety  can be obtained analogously.
Given NPs, we can compute  \emph{sufficient preconditions} for safety and unsafety as follows. $
\suffpre{s}{(P)}=\necpre{s}(P) \wedge  \neg \necpre{u}(P)$ and similarly $\suffpre{u}{(P)}=\necpre{u}(P) \wedge  \neg \necpre{s}(P).
$ A  precondition is \emph{optimal} if  $\necpre{s}(P) \wedge   \necpre{u}(P)$ is \emph{unsatisfiable}. Then we have $\suffpre{s}{(P)}=\necpre{s}(P)$, that is,  the necessary and the sufficient conditions are the same for the validity of the assertion. Next, we show how a given set of SPs can be combined to derive a new SPs for a program. 

\begin{proposition}[Composing  Preconditions]
\label{prop:disj_precond}
Let $\phi_1, \phi_2, \ldots $ be a (possibly infinite) set of formulas such that each $\phi_i$ is an SP for (un)safety of  $P$. Then their disjunction $\bigvee_i \phi_i$ is also an SP  for (un)safety of $P$. 
\end{proposition}
\begin{proof}[sketch]
Consider the case for safety, that is, each $\phi_i$ ($i=1,2,\ldots$) is an SP for safety of $P$. Let $P'_i=\replace(P, \phi_i)$.  For each $\phi_i$, we have $P'_i \not \vdash_\theory \unsafe$. That is,  \unsafe~is not derivable from any of $P'_i$. Let $P'=\replace(P, {\bigvee_i \phi_i})$. The initial clauses of $P'$ are the union of all the initial clauses of each $P'_i$ while the rest of the clauses are the same. From our assumption, any derivation of \unsafe~need to use at least one of the initial clauses, our hypothesis under this condition ensures that such a derivation  either does not exist or is  infeasible.
Then obviously,  \unsafe~is not derivable from $P'$. The proof for unsafety can be done analogously. \qed
\end{proof}
%Similarly, the sufficient precondition for unsafety is given by \[ \suffpre{u}{(P)}=\necpre{u}(P) \wedge  \neg \necpre{s}(P). \]
Next, we present a special kind of under-approximation of CHCs by restricting their initial states which is useful in partitioning the set of CHCs.  We say a set of CHCs $P'$ is an under-approximation of $P$ if for an atom $A$, we have $P' \vdash_{\theory} A$ entails 
$P \vdash_{\theory} A$ (the reverse is not necessarily true).
\begin{definition}[Under-approximating CHCs by restricting initial states]\label{def:ris}
Let $P$ be a set of CHCs and $\phi$ be a constraint over $\theory$. Let $P'$ be the set of 
clauses obtained from $P$ by replacing the initial clauses 
$\{(p^I(\tuplevar{x}) \leftarrow \theta_i) \mid 1 \le i \le k\}$ by 
$\{(p^I(\tuplevar{x}) \leftarrow  \theta_i \wedge \phi) \mid 1 \le i \le k\}$.   This operation is denoted by $\ris(P,\phi)$.
\end{definition}
\begin{proposition}
\label{prop:under-approx}
Let $P$ be a set of CHCs, $\phi$ be any constraint over $\theory$ and $P'=\ris(P,\phi)$. Then $P'$ is an under-approximation of $P$.
\end{proposition}
\begin{proof}[sketch]
This is obvious since $P'$ is initial clauses restricted version of $P$. \qed
\end{proof}

Note that $\ris(P,\phi)$ only restricts the initial clauses of $P$ while keeping others intact. Then Definition \ref{def:nec-suff-precond} implies that an SP of $\ris(P,\phi)$ is also an SP of $P$. Formally,  
\begin{proposition}
\label{prop:lift_precond}
Let $P$ be a set of CHCs, $\phi$ be any constraint over $\theory$ and $P'=\ris(P,\phi)$. Let $\psi$ be an SP for (un)safety of $P'$ then $\psi$ is also a  SP for (un)safety of $P$.
\end{proposition}
\begin{proof}[sketch]
Consider the case for safety. Let $\psi$ be an SP for safety of $P'$. Since $P$ and $P'$ differ only in the set of initial clauses, both the constructions $\replace(P,\phi)$ and $\replace(P',\phi)$ yield the same set of clauses. So $\psi$ is an SP for safety of $P$ as well. The proof for the case of unsafety can be done analogously. \qed
\end{proof}
%The necessary preconditions thus derived may be too imprecise and the corresponding sufficient preconditons may be sufficiently strong making their use limited. 
We now present some CHC transformations, taken from the literature on CLP and Horn clause verification. These are well known transformations and we refer to ~\cite{kafle-iclp18} for details.

\subsubsection{Partial Evaluation (PE).}
\label{pe}

Partial evaluation  \cite{Jones-Gomard-Sestoft} with respect to a goal specialises a program for the given goal; preserving only those derivations that are relevant for
deriving the goal. PE algorithm we apply here \cite{gallagher:pepm93,kafle-iclp18} produces a polyvariant specialisation, that is, a finite number of versions of each predicate, which is essential for deriving disjunctive information. 

The result of applying PE to the example program in Figure \ref{ex:precond} 
with respect to \unsafe~and \safe~is shown in Figure~\ref{pe_example}. The algorithm is a bit involved, so we do not provide details here and ask the readers to refer to \cite{kafle-iclp18}. However a key point to note is that due to polyvariant specialisation the predicates \emph{init} and \emph{wh} are split, leading to a more precise approximations as shown below. But all the feasible derivations of \safe~and \unsafe~are preserved.
\begin{figure}[t]
\centerline{
  \begin{tabular}{ll}
    \pcode[\small]{
$\mathtt{unsafe \leftarrow B<0,A\leq0,wh\_2(A,B)}$. \\
$\mathtt{wh\_2(A,B) \leftarrow B<0,A \leq 0,init\_2(A,B)}$. \\
$\mathtt{wh\_2(A,B) \leftarrow B<0,A=0,C=1,}$ \\
~~~~~~~~~~~~~~~~~~~$\mathtt{B-D= -1,wh\_1(C,D)}$. \\
$\mathtt{wh\_1(A,B) \leftarrow A \geq 1,init\_1(A,B)}$. \\
$\mathtt{wh\_1(A,B) \leftarrow A \geq 1,A-C= -1,}$ \\
~~~~~~~~~~~~~~~~~~~$\mathtt{B-D= -1,wh\_1(C,D)}$. \\
$\mathtt{init\_1(A,B) \leftarrow A \ge 1}.$ \\ 
$\mathtt{init\_2(A,B) \leftarrow A \le 0, B \ge 0}.$
    } &
    \pcode[\small]{
$\mathtt{safe \leftarrow B \geq 0,A\leq0,wh\_2(A,B)}$. \\
$\mathtt{wh\_2(A,B) \leftarrow B \geq 0,A \leq 0,init\_2(A,B)}$. \\
$\mathtt{wh\_2(A,B) \leftarrow B \geq 0,A=0,C=1,}$ \\
~~~~~~~~~~~~~~~~~~~$\mathtt{B-D= -1,wh\_1(C,D)}$. \\
$\mathtt{wh\_1(A,B) \leftarrow A \geq 1,B \geq 0,init\_1(A,B)}$. \\
$\mathtt{wh\_1(A,B) \leftarrow B \geq 0,A \geq 1,A-C= -1,}$ \\
~~~~~~~~~~~~~~~~~~~$\mathtt{B-D= -1,wh\_1(C,D)}$. \\
$\mathtt{init\_1(A,B) \leftarrow A \geq 1,B \geq 0}.$ \\ 
$\mathtt{init\_2(A,B) \leftarrow B \geq 0,A \leq 0}.$
} 
  \end{tabular}
% \end{center}
}
\caption{Partially evaluated programs: wrt \texttt{unsafe} (left) and wrt \texttt{safe} (right).\label{pe_example}}
\end{figure}
We derive the following NPs from these programs: 
\[
\begin{array}{lcl}
   \necpre{s} \equiv \mathtt{(B\ge0 \wedge A\leq 0) \vee 
	(A\geq 1 \wedge B\geq 0) \equiv B\geq 0}. 
\\ \necpre{u} \equiv \mathtt{(B<0 \wedge A\leq 0) \vee 
	A\geq 1 \equiv B<0 \vee A\geq 1}.
\end{array}
\]
%The following preconditions  are derived from the programs in Figure \ref{pe_example} following Definition \ref{def:necsafe}. 
%\begin{flalign*}
%\mathtt{\necpre{s} \equiv (B\ge0 \wedge A\leq 0) \vee (A\geq 1 \wedge B\geq 0) \equiv B\geq 0.} \\
%\mathtt{\necpre{u} \equiv (B<0 \wedge A\leq 0) \vee A\geq 1 \equiv B<0 \vee A\geq 1.}
%\end{flalign*}

\subsubsection{Constraint Specialisation (CS) \cite{DBLP:journals/scp/KafleG17}.}
\label{cs}

CS strengthens constraints in the clauses, while preserving
derivations of a given atom (goal). The transformation prunes those paths that are not relevant for deriving the given atom. Next we explain the transformation with an example.

%\ignore{
% Formally, we define it as follows. 
%\begin{definition}[Constraint specialisation]\label{def:cs}
%A constraint specialisation of $P$ with respect to a goal $A$ is a transformation in which 
% each constraint $\phi$ in a clause of $P$ is replaced by a constraint $\phi \wedge \psi$
% such that the resulting set of clauses preserves the derivation of $A$.
% \end{definition}

%In our experiments, the constraint $\psi$ above is obtained by applying 
% abstract interpretation  to a query-answer transformed version of the 
% set of CHCs.  The method is described in detail in \cite{DBLP:journals/scp/KafleG17}.

Figure \ref{cs_tiny_example} (left) shows an example program and Figure \ref{cs_tiny_example} (right) its constraint specialised version preserving the derivation of \unsafe. The strengthened constraints are obtained by recursively propagating $\mathtt{B \ge A}$ top-down from the goal $\mathtt{unsafe}$ and $\mathtt{A = 1, B=1}$ bottom-up from the constrained fact. An invariant $\mathtt{A \ge B, B \ge 0}$ for the recursive predicate $\texttt{p(A,B)}$ in derivations of $\unsafe$ is computed and conjoined to each call to $\texttt{p}$ in the clauses (underlined in the clauses in Figure \ref{cs_example} (right)).

\newcommand{\myleftarrow}{\leftarrow}

\begin{figure}[h]
\centerline{
% \begin{center}
  \begin{tabular}{ll}
    \pcode[\small]{
$\mathtt{unsafe \myleftarrow B>A, p(A,B)}$. \\
$\mathtt{p(A+B,B+1) \myleftarrow  p(A,B)}$.\\
$\mathtt{p(A,B) \myleftarrow A=1, B=0}$.
    } &
    \pcode[\small]{
    $\mathtt{unsafe \myleftarrow B>A, \underline{B\geq 0, A \geq B}, p(A,B)}$. \\
$\mathtt{p(A+B,B+1) \myleftarrow \underline{B\geq 0, A \geq B}, p(A,B)}$.\\
$\mathtt{p(A,B) \myleftarrow A=1, B=0,\underline{B\geq 0, A \geq B}}$.
} 
  \end{tabular}
% \end{center}
}
\caption{Example program (left) and its constraint specialised version preserving the derivation of \unsafe~ (right).\label{cs_tiny_example}}
\end{figure}
\noindent
%}

The result of applying CS to the program in Figure \ref{pe_example} (left) with respect to \unsafe~and Figure \ref{pe_example} (right) with respect to \safe~is  shown in Figure \ref{cs_example}.
\begin{figure}[b]
\centerline{
  \begin{tabular}{ll}
    \pcode[\small]{
$\mathtt{unsafe \leftarrow  B \le 0,  A<0,  wh\_2(B,A)}$. \\ 
$\mathtt{wh\_2(A,B) \leftarrow  B<0,  A \le 0, init\_2(A,B)}$. \\ 
$\mathtt{wh\_2(A,B) \leftarrow B<0,  A=0,  C=1,}$  \\
~~~~~~~~~~~~~~~~$\mathtt{B-D= -1, wh\_1(C,D)}$. \\ 
$\mathtt{wh\_1(A,B) \leftarrow  A>B,  A \ge 1, init\_1(A,B)}$. \\ 
$\mathtt{wh\_1(A,B) \leftarrow  A>B,  A \ge 1,  A-C= -1, }$  \\
~~~~~~~~~~~~~~~~$\mathtt{ B-D= -1, wh\_1(C,D)}$.  \\
$\mathtt{init\_1(A,B) \leftarrow A>B,  A \ge 1}.$ \\
 $\mathtt{init\_2(A,B) \leftarrow B < 0,A \leq 0}.$ 
    } &
    \pcode[\small]{
$\mathtt{safe \leftarrow B \ge 0,A \le 0,wh\_2(A,B)}$. \\ 
$\mathtt{wh\_2(A,B) \leftarrow B \ge 0,A \le 0,init\_2(A,B)}$. \\ 
$\mathtt{wh\_2(A,B) \leftarrow B \ge 0,A=0,C=1,}$ \\
~~~~~~~~~~~~~~~~$\mathtt{B-D= -1,wh\_1(C,D)}$. \\ 
$\mathtt{wh\_1(A,B) \leftarrow A \ge 1,B \ge A,init\_1(A,B)}$. \\ 
$\mathtt{wh\_1(A,B) \leftarrow B \ge 0,A \ge 1,A-C= -1,}$ \\
~~~~~~~~~~~~~~~~$\mathtt{B-D= -1,wh\_1(C,D)}$. \\
$\mathtt{init\_1(A,B) \leftarrow A \ge 1,B \ge A}.$ \\
 $\mathtt{init\_2(A,B) \leftarrow B \geq 0,A \leq 0}.$
} 
  \end{tabular}
}
\caption{Constraint specialised programs: wrt \texttt{unsafe} (left) and wrt \texttt{safe} (right)\label{cs_example}}
\end{figure}
Then we derive the following NPs: $\necpre{s} \equiv \mathtt{(B\ge0 \wedge A\leq 0) \vee (A\geq 1 \wedge B\geq A)}$ and $\necpre{u} \equiv \mathtt{(B<0 \wedge A\leq 0) \vee (A\geq 1 \wedge A> B)}$.
% from the programs in Figure \ref{cs_example}. 
%\begin{flalign*}
%\mathtt{\necpre{s} \equiv (B\ge0 \wedge A\leq 0) \vee (A\geq 1 \wedge B\geq A)} ~~
%\mathtt{\necpre{u} \equiv (B<0 \wedge A\leq 0) \vee (A\geq 1 \wedge A> B).}
%\end{flalign*}
Since $\necpre{s} \wedge \necpre{u}$ is \emph{unsatisfiable}, these preconditions are optimal.

\subsubsection{Trace Elimination (TE).}
\label{te}
TE is a program transformation that eliminates an AND-tree of $P$ from $P$ while preserving the rest of its AND-trees.  We illustrate the transformation via the program in Figure \ref{te_example}. The program on the right is the result of eliminating the AND-tree $c_1-c_3$ from the program on the left. Note that the elimination has caused splitting of the predicate $p$ so that the same AND-tree is no longer possible to construct while the rest of the AND-trees are preserved. In summary, the method consists of representing the program and the trace as finite tree automata, removing the trace using automata difference construction and reconstructing a program from the result as detailed in \cite{DBLP:journals/cl/KafleG17}.
\begin{figure}[H]
\centerline{
% \begin{center}
  \begin{tabular}{ll}
    \pcode[\small]{
$\mathtt{c_1. ~unsafe \myleftarrow B>A, p(A,B)}$. \\
$\mathtt{c_2.~ p(A+B,B+1) \myleftarrow  p(A,B)}$.\\
$\mathtt{c_3. ~p(A,B) \myleftarrow A=1, B=0}$.
    } &
    \pcode[\small]{
   $\mathtt{unsafe \myleftarrow B>A, p\_1(A,B)}$. \\
   $\mathtt{p\_1(A+B,B+1) \myleftarrow  p(A,B)}$.\\
$\mathtt{p(A+B,B+1) \myleftarrow  p(A,B)}$.\\
$\mathtt{p(A,B) \myleftarrow A=1, B=0}$.
} 
  \end{tabular}
% \end{center}
}
\caption{A program (left) and its version after removing the trace $c_1-c_3$ (right).\label{te_example}}
\end{figure}
TE is a useful transformation since it can act as a refinement or problem decomposition as we shall see later. Furthermore, the splitting of the  predicates can help derive disjunctive information. While the elimination of infeasible trees does not have any effect on preconditions, extra care must be taken while eliminating feasible ones. Lemma \ref{lemma:feasible} allows us to derive a safe precondition in this case. To be consistent with the rest of the transformations, we can define TE which takes a program $P$ and an atom $A$ and searches for a derivation of $P$ rooted at $A$ and eliminates it if any.

\begin{definition}[$\theta_t$]
\label{def:trace_cs_init}
Let $P$ be a set of CHCs and $t$ be a \emph{feasible} AND-tree  from $P$ for (un)safe. 
Let $p^I(\tuplevar{x})$ be the atom label of an initial node of $t$ then $\theta_t = \constr(t)\vert_\tuplevar{x}$. 
\end{definition}
In other words, $\theta_t$ is basically the result of  eliminating all  variables but $\tuplevar{x}$ from $\constr(t)$. Note that  $\theta_t$ is necessary condition  for $t$ to be feasible.  

\begin{lemma}[\cite{kafle-iclp18}]\label{lemma:feasible}
Let $P'$ be the result of eliminating a \emph{feasible} AND-tree $t$ for $\mathtt{safe}~(\mathtt{unsafe})$ from $P$. 
%Let $p^I(\tuplevar{x})$ be the atom label of an initial node of $t$ and let $\theta = \constr(t)\vert_\tuplevar{x}$.
Then $\necpre{s}(P) =\necpre{s}(P') \vee \theta_t$ ($\necpre{u}(P) =\necpre{u}(P') \vee \theta_t$), where $\theta_t$ is a constraint extracted from $t$ (Def. \ref{def:trace_cs_init}).
\end{lemma}
Transformations such as PE, CS and TE (which removes \emph{infeasible} trees)
not only preserve the goal but also the initial clauses. This allows us  to
construct a sequence of clauses  $P_0,P_1,\ldots,P_m$ where $P = P_0$ and each
element of the sequence is more specialised than its predecessor with respect to
derivations of  \safe~($\unsafe$).  As a consequence, the NPs are more precise.
We write  $P \Longrightarrow_A P'$ when $P'$ is a goal-preserving transformation of $P$ with respect to an atom $A$, that is, $P \models A ~\text{iff}~ P' \models A$. TE (eliminating \emph{feasible} trees) is a little different which does not preserve the goal. We abuse the notation and write $P \Longrightarrow_{t_A} P'$ for transformation of $P$ by eliminating a feasible tree rooted at $A$, yielding $P'$. Lemma~\ref{lemma:feasible} shows that soundness is ensured in this case. 

Let us now wrap these transformations and their combinations.
Let \tr\ and \trseq\ be any functions satisfying the following:
%which take a program, an atom and a constraint over $\theory$ and returns a tuple consisting of a program and a constraint. 
\begin{align*}
\trA \tuple{P,\phi} &= \left\{
\begin{array}{cl}
			\tuple{P', \phi} & \text{where~} P \Longrightarrow_A P' \\
      \textit{or} \\
			\tuple{P', \phi'} & \text{where~} P \Longrightarrow_{t_A} P' \text{and}~ \phi'=\phi \vee \theta_{t_A} \text{(Def.\ \ref{def:trace_cs_init})}
			\end{array} \right.
\\[1ex]
\trseqA \tuple{P, \phi, n} &= \trA^n \tuple{P, \phi} \text{~for~} n\geq 1
\end{align*}
%\[ \text{\trseq}(P, A, \phi) = \begin{cases}   
%			\tr(P, A, \phi) \\
%			\text{\trseq}(P', A, \phi') ~\text{where}~ \tuple{P', \phi'} =  \tr(P, A, \phi).
%			\end{cases} 
%\]
where $f^1 = f$ and $f^n = f^{n-1} \circ f$.

\trseq\ allows us to combine the above transformations in any order and Proposition \ref{prop:sequence} allows us to derive more precise preconditions from the combination.
\begin{proposition}[Adapted from Proposition 2 of \cite{kafle-iclp18}]\label{prop:sequence}
Let $P$ be a set of CHCs and $n\geq 1$. Let $\tuple{P_s,\phi_s} = \trseqA[\safe] \tuple{P,\falseit, n}$.
Then $ \models_{\theory} (\necpre{s}{(P_s)} \vee \phi_s) \rightarrow \necpre{s}{(P)}$.
Similarly, if $\tuple{P_u, \phi_u} = \trseqA[\unsafe] \tuple{P,\falseit, n}$,
then $ \models_{\theory} (\necpre{u}{(P_u)} \vee \phi_u) \rightarrow \necpre{u}{(P)}$.
\end{proposition}
%These transformations can be combined in any order to derive more precise  preconditions as stated in Proposition \ref{prop:sequence}. 
%\begin{proposition}[Adapted from \cite{kafle-iclp18}]\label{prop:sequence}
%Let $P = P_0$, $\psi_0 = \falseit$ and let   $(P_0, \psi_0)$, $(P_1, \psi_1 ),\ldots,(P_m, \psi_m)$ be a sequence of pairs where 
%for $(0 \le i < m)$
%either (i) $P_i \Longrightarrow_{\safe} P_{i+1}$ and $\psi_{i} = \psi_{i+1}$, or 
%(ii) $P_i \Longrightarrow_{t_{\safe}} P_{i+1}$ where $t_{\safe}$ is feasible AND-tree rooted at \safe and
%$\psi_{i+1} = \psi_i \vee \theta$, where $\theta$ is a constraint extracted from $t_{\safe}$ (Def. \ref{def:trace_cs_init}).
%Then $ \models_{\theory} (\necpre{s}{(P_m)} \vee \psi_m) \rightarrow \necpre{s}{(P)}$.
%\end{proposition}
%As we have shown, applying partial evaluation followed by constraint 
%specialisation for our running example was sufficient to derive an 
%optimal  precondition. 
%However, in more complex cases we need a more sophisticated algorithm 
% to detect optimality or derive non-trivial preconditions.

\section{An algorithm for precondition inference}\label{sec:precond}
Given a set of CHC transformations  and their relationships with the preconditions, we now state the algorithm for computing SPs for safety  and unsafety  in Algorithm \ref{alg:precond}.

\begin{algorithm}
\caption{Algorithm for inferring sufficient preconditions}
\label{alg:precond}
\begin{algorithmic}[1]
\State \textbf{Input}: A set of CHCs $P$ with clauses for \safe~and \unsafe and the length of the sequence of transformations $n$.
\State \textbf{Output:} Sufficient precondition for safety and unsafety wrt  \safe~and \unsafe.
\State \textbf{Initialisation:} $\suffpre{s} \gets \falseit$; $\suffpre{u} \gets \falseit$; $P_s \gets P$; $P_u \gets P$; \\
 ~~~~~~~~~~~~~~~~~~~~$\phi_{old} \gets \necpre{s}{(P)}$ (Definition \ref{def:necsafe});
\While{true}
\State
\begin{tabular}{|l||l|}
\hline
 $\tuple{P_s, \theta_s} \gets \trseqA[\safe]\tuple{P_s,\falseit, n}$; ~~~~&  ~~~~$\tuple{P_u, \theta_u} \gets \trseqA[\unsafe]\tuple{P_u,\falseit, n}$; \\ 
 $\phi_s \gets \necpre{s}{(P_s)} \vee \theta_s$ ~~~~& ~~~~$\phi_u \gets \necpre{u}{(P_u)} \vee \theta_u$ \\ \hline
\end{tabular}
\State $\phi_{new}\gets \phi_s \wedge \phi_u$
\If{$\phi_{new}\equiv \falseit$ (optimal condition reached)} 
\State
\begin{tabular}{|l||l|}
\hline
 $\suffpre{s} \gets \suffpre{s} \vee \phi_s$ ~~~~& ~~~~  $\suffpre{u} \gets \suffpre{u} \vee \phi_u$ \\ \hline
 \end{tabular}
\State \Return $\tuple{\suffpre{s},{\suffpre{u}}}$

\ElsIf{$\phi_{old} \rightarrow \phi_{new}$ (approximation does not get stronger)}

\State 
\begin{tabular}{|l||l|}
\hline 
$\suffpre{s} \gets \suffpre{s} \vee (\phi_s \wedge \neg \phi_u )$ ~~~~& ~~~~ $\suffpre{u} \gets \suffpre{u} \vee (\phi_u \wedge \neg \phi_s)$ \\ \hline
\end{tabular}
\State \Return $\tuple{\suffpre{s},{\suffpre{u}}}$

\Else ~(refine clauses via under-approximation using Def. \ref{def:ris})
\State $\phi_{old} \gets \phi_{new}$
\State 
\begin{tabular}{|l||l|}
\hline
$\suffpre{s} \gets \suffpre{s} \vee (\phi_s \wedge \neg \phi_u )$; ~~~~&  ~~~~ $\suffpre{u} \gets \suffpre{u} \vee (\phi_u \wedge \neg \phi_s)$; \\
$P_s \gets \ris(P_s, \phi_{new})$~~~~&~~~~ $P_u \gets \ris(P_u, \phi_{new})$ \\ \hline
\end{tabular}
\EndIf
\EndWhile
\end{algorithmic}
\end{algorithm}
The algorithm takes a set of CHCs as input (with special predicates encoding the sets of safe  and  unsafe states) and a positive number $n$ representing the length of a sequence of transfomations and returns a pair of SPs for safety and unsafety. The SPs $\suffpre{s}$ and $\suffpre{u}$ are initialised to $\falseit$. The algorithm aims to weaken these SPs as much as possible. $\phi_{old}$ keeps track of the set of initial states that are yet to be proven safe or unsafe and is initialised to the set of initial states of $P$. $P_s$ and $P_u$ keep track of the transformations of $P$ with respect to \safe~and \unsafe~respectively, and also their under-approximations.

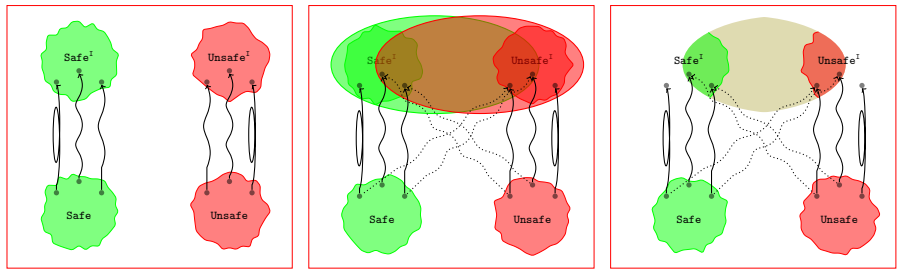
\begin{figure}[ht!]
\centering
\begin{minipage}{.32\textwidth}
\tikz \node[scale=0.5, inner sep=0]{
\begin{tikzpicture}
[align=center,node distance=4cm and 6cm]
  \tikzstyle{ellipse} = [ellipse, minimum height=8em, minimum width=5em, draw,thick]
\tikzset{dot/.style={circle,fill=#1,inner sep=0,minimum size=4pt}}
  \tikzstyle{block} = [rectangle, rounded corners, minimum height=1em, minimum width=6em, draw]
 
\coordinate (safeC) at (-3,-3);
\coordinate[right of=safeC]  (unsafeC) ;
\coordinate[above of=safeC] (isafeC) ;
\coordinate[above of=unsafeC] (iunsafeC) ;
  
  \matrix {

  \begin{scope}[fill opacity=0.5]
  %drawing concrete set of inital safe states, inital unsafe states, safe states, unsafe states
   \draw[green, fill,rounded corners=.5mm, name=safe]  (safeC) \irregularcircle{1cm}{1mm} {} ;
   \node[below,black, opacity=1.0] at (safeC.south){$\mathtt{Safe}$};
    \draw[red,fill, rounded corners=.5mm, name=unsafe]  (unsafeC) \irregularcircle{1cm}{1mm} {};
    \node[below,black, opacity=1.0] at (unsafeC.south){$\mathtt{Unsafe}$};
    \draw[green,fill,rounded corners=.5mm, name=isafe]  (isafeC) \irregularcircle{1cm}{1mm} {} ;
    \node[above,black, opacity=1.0] at (isafeC.south){$\mathtt{Safe^I}$};
    \draw[red,fill, rounded corners=.5mm, name=iunsafe]  (iunsafeC) \irregularcircle{1cm}{1mm} {};
     \node[above,black, opacity=1.0] at (iunsafeC.south){$\mathtt{Unsafe^I}$};
    
   %drawing points inside the concrete states
   \node [dot] (s1) at (-3,-2.2) {};
  \node[dot] (s2) at (-3.6,-2.5){};
  \node[dot] (s3) at (-2.4,-2.5){};
  
  \node [dot] (is1) [above of=s1, yshift=-30pt]{};
  \node[dot] (is2) [above of=s2, yshift=-30pt]{};
  \node[dot] (is3) [above of=s3, yshift=-30pt]{};
  
  \node [dot] (iu1) [right of=is1]{};
  \node[dot] (iu2) [right of=is2]{};
  \node[dot] (iu3) [right of=is3]{};
  
    \node [dot] (u1) [right of=s1]{};
  \node[dot] (u2) [right of=s2]{};
  \node[dot] (u3) [right of=s3]{};
     %drawing reachability arrows between the concrete set of states
      \draw [<-, thick,snake=coil,segment aspect=10,segment length=17pt] (is2) -- (s2)  {} ;
       \draw [<-, thick,snake=coil,segment length=30pt] (is1) -- (s1)  {} ;
       \draw [<-, thick,snake=coil,segment length=40pt] (is3) -- (s3)  {} ;
       
           \draw [<-, thick,snake=coil,segment length=30pt] (iu1) -- (u1) node[midway, left] {} ;
      \draw [<-, thick,snake=coil,segment length=40pt] (iu2) -- (u2) node[midway, left] {} ;
       \draw [<-, thick,snake=coil,segment aspect=10,segment length=17pt] (iu3) -- (u3) node[midway, left] {} ;

    \end{scope}
\\};

\begin{pgfonlayer}{background}
\draw [draw=red]
([xshift=-27pt,yshift=-15pt]safeC.south)

rectangle ([xshift=75pt,yshift=70pt]iunsafeC.north);
\end{pgfonlayer}
\end{tikzpicture}
};
\end{minipage}
\begin{minipage}{.32\textwidth}
\tikz \node[scale=0.5, inner sep=0]{
\begin{tikzpicture}
[align=center,node distance=4cm and 6cm]
  \tikzstyle{sharp} = [ellipse, minimum height=8em, minimum width=17em, draw,thick, fill]
\tikzset{dot/.style={circle,fill=#1,inner sep=0,minimum size=4pt}}
  \tikzstyle{block} = [rectangle, rounded corners, minimum height=1em, minimum width=6em, draw]

\coordinate (safeC) at (-3,-3);
\coordinate[right of=safeC]  (unsafeC) ;
\coordinate[above of=safeC] (isafeC) ;
\coordinate[above of=unsafeC] (iunsafeC) ;
  
  \matrix {
  
  \begin{scope}[fill opacity=0.5]
  %drawing concrete set of inital safe states, inital unsafe states, safe states, unsafe states
   \draw[green, fill,rounded corners=.5mm, name=safe]  (safeC) \irregularcircle{1cm}{1mm} {} ;
   \node[below,black, opacity=1.0] at (safeC.south){$\mathtt{Safe}$};
    \draw[red,fill, rounded corners=.5mm, name=unsafe]  (unsafeC) \irregularcircle{1cm}{1mm} {};
    \node[below,black, opacity=1.0] at (unsafeC.south){$\mathtt{Unsafe}$};
    \draw[green,fill,rounded corners=.5mm, name=isafe]  (isafeC) \irregularcircle{1cm}{1mm} {} ;
    \node[above,black, opacity=1.0] at (isafeC.south){$\mathtt{Safe^I}$};
    \node [sharp, draw, green] (safesharp) at ([xshift=40pt]isafeC){};
    \draw[red,fill, rounded corners=.5mm, name=iunsafe]  (iunsafeC) \irregularcircle{1cm}{1mm} {};
     \node[above,black, opacity=1.0] at (iunsafeC){$\mathtt{Unsafe^I}$};
     \node [sharp, draw, red] (unsafesharp) at ([xshift=-40pt]iunsafeC){};
    
   %drawing points inside the concrete states
   \node [dot] (s1) at (-3,-2.2) {};
  \node[dot] (s2) at (-3.6,-2.5){};
  \node[dot] (s3) at (-2.4,-2.5){};
  
  \node [dot] (is1) [above of=s1, yshift=-30pt]{};
  \node[dot] (is2) [above of=s2, yshift=-30pt]{};
  \node[dot] (is3) [above of=s3, yshift=-30pt]{};
  
  \node [dot] (iu1) [right of=is1]{};
  \node[dot] (iu2) [right of=is2]{};
  \node[dot] (iu3) [right of=is3]{};
  
    \node [dot] (u1) [right of=s1]{};
  \node[dot] (u2) [right of=s2]{};
  \node[dot] (u3) [right of=s3]{};
     %drawing reachability arrows between the concrete set of states
      \draw [<-, thick,snake=coil,segment aspect=10,segment length=17pt] (is2) -- (s2)  {} ;
       \draw [<-, thick,snake=coil,segment length=30pt] (is1) -- (s1)  {} ;
       \draw [<-, thick,snake=coil,segment length=40pt] (is3) -- (s3)  {} ;
       
           \draw [<-, thick,snake=coil,segment length=30pt] (iu1) -- (u1) node[midway, left] {} ;
      \draw [<-, thick,snake=coil,segment length=40pt] (iu2) -- (u2) node[midway, left] {} ;
       \draw [<-, thick,snake=coil,segment aspect=10,segment length=17pt] (iu3) -- (u3) node[midway, left] {} ; 
       
  %draw dotted arrows to the wrong concrete state
        \draw [<-,dotted, thick,snake=coil,segment length=30pt] (is1) -- (u2)  {} ;
        \draw [<-,dotted, thick,snake=coil,segment length=40pt] (is3) -- (u1)  {} ;
        \draw [<-,dotted, thick,snake=coil,segment length=30pt] (iu1) -- (s3)  {} ;
        \draw [<-,dotted, thick,snake=coil,segment length=40pt] (iu2) -- (s2)  {} ;
     
    \end{scope}
\\};

\begin{pgfonlayer}{background}
\draw [draw=red]
([xshift=-27pt,yshift=-15pt]safeC.south)

rectangle ([xshift=75pt,yshift=70pt]iunsafeC.north);
\end{pgfonlayer}
\end{tikzpicture}
};
\end{minipage}
\begin{minipage}{.32\textwidth}

\tikz \node[scale=0.5, inner sep=0]{
\begin{tikzpicture}
[align=center,node distance=4cm and 6cm]
  \tikzstyle{sharp} = [ellipse, minimum height=8em, minimum width=17em, draw,thick, fill]
\tikzset{dot/.style={circle,fill=#1,inner sep=0,minimum size=4pt}}
  \tikzstyle{block} = [rectangle, rounded corners, minimum height=1em, minimum width=6em, draw]

\coordinate (safeC) at (-3,-3);
\coordinate[right of=safeC]  (unsafeC) ;
\coordinate[above of=safeC] (isafeC) ;
\coordinate[above of=unsafeC] (iunsafeC) ;
  
  \matrix {
  
  \begin{scope}[fill opacity=0.5]
  %drawing concrete set of inital safe states, inital unsafe states, safe states, unsafe states
   %\draw[green, fill,rounded corners=.5mm, name=safe]  (safeC) \irregularcircle{1cm}{1mm} {} ;
   \draw[green, fill,rounded corners=.5mm, name=safe]  (safeC) \irregularcircle{1cm}{1mm} {} ;
   \node[below,black, opacity=1.0] at (safeC.south){$\mathtt{Safe}$};
    \draw[red,fill, rounded corners=.5mm, name=unsafe]  (unsafeC) \irregularcircle{1cm}{1mm} {};
    \node[below,black, opacity=1.0] at (unsafeC.south){$\mathtt{Unsafe}$};
    %\draw[green,fill, rounded corners=.5mm, name=isafe]  (isafeC) \irregularcircle{1cm}{1mm} {} ;
    %\node[above,black, opacity=1.0] at (isafeC.south){$\mathtt{Safe^I}$};
    %\draw[red, fill, rounded corners=.5mm, name=iunsafe]  (iunsafeC) \irregularcircle{1cm}{1mm} {};
    % \node[above,black, opacity=1.0] at (iunsafeC){$\mathtt{Unsafe^I}$};
     %\begin{pgfonlayer}{background}
    \begin{scope}
    \path [clip] ([xshift=40pt]isafeC) ellipse [x radius=8.5em, y radius=4em];
    \path [clip] ([xshift=-40pt]iunsafeC) ellipse [x radius=8.5em, y radius=4em];
    %\path [draw,fill,green] ([xshift=40pt]isafeC) ellipse [x radius=8.5em, y radius=4em];
    %\path [draw,fill,red] ([xshift=-40pt]iunsafeC) ellipse [x radius=8.5em, y radius=4em];
%    \node[minimum height=8em, minimum width=17em,clip] at ([xshift=40pt]isafeC) {};
%    \node[minimum height=8em, minimum width=17em,clip] at ([xshift=-40pt]iunsafeC) {};
%     \path [sharp,clip]  at ([xshift=-40pt]iunsafeC);
%     \node [sharp, draw, green] (safesharp) at ([xshift=40pt]isafeC){};
%     \node [sharp, draw, red] (unsafesharp) at ([xshift=-40pt]iunsafeC){};
      \draw [fill,olive!60] ($(isafeC)!0.5!(iunsafeC)$) circle (20em);
    \draw[green,fill, rounded corners=.5mm, name=isafe]  (isafeC) \irregularcircle{1cm}{1mm} {} ;
    \draw[red, fill, rounded corners=.5mm, name=iunsafe]  (iunsafeC) \irregularcircle{1cm}{1mm} {};
     \end{scope}
    \node[above,black, opacity=1.0] at (isafeC.south){$\mathtt{Safe^I}$};
     \node[above,black, opacity=1.0] at (iunsafeC){$\mathtt{Unsafe^I}$};
     %\end{pgfonlayer}
    
   %drawing points inside the concrete states
   \node [dot] (s1) at (-3,-2.2) {};
  \node[dot] (s2) at (-3.6,-2.5){};
  \node[dot] (s3) at (-2.4,-2.5){};
  
  \node [dot] (is1) [above of=s1, yshift=-30pt]{};
  \node[dot] (is2) [above of=s2, yshift=-30pt]{};
  \node[dot] (is3) [above of=s3, yshift=-30pt]{};
  
  \node [dot] (iu1) [right of=is1]{};
  \node[dot] (iu2) [right of=is2]{};
  \node[dot] (iu3) [right of=is3]{};
  
    \node [dot] (u1) [right of=s1]{};
  \node[dot] (u2) [right of=s2]{};
  \node[dot] (u3) [right of=s3]{};
     %drawing reachability arrows between the concrete set of states
      \draw [<-, thick,snake=coil,segment aspect=10,segment length=17pt] (is2) -- (s2)  {} ;
       \draw [<-, thick,snake=coil,segment length=30pt] (is1) -- (s1)  {} ;
       \draw [<-, thick,snake=coil,segment length=40pt] (is3) -- (s3)  {} ;
       
           \draw [<-, thick,snake=coil,segment length=30pt] (iu1) -- (u1) node[midway, left] {} ;
      \draw [<-, thick,snake=coil,segment length=40pt] (iu2) -- (u2) node[midway, left] {} ;
       \draw [<-, thick,snake=coil,segment aspect=10,segment length=17pt] (iu3) -- (u3) node[midway, left] {} ; 
       
  %draw dotted arrows to the wrong concrete state
        \draw [<-,dotted, thick,snake=coil,segment length=30pt] (is1) -- (u2)  {} ;
        \draw [<-,dotted, thick,snake=coil,segment length=40pt] (is3) -- (u1)  {} ;
        \draw [<-,dotted, thick,snake=coil,segment length=30pt] (iu1) -- (s3)  {} ;
        \draw [<-,dotted, thick,snake=coil,segment length=40pt] (iu2) -- (s2)  {} ;
     
    \end{scope}
\\};

\begin{pgfonlayer}{background}
\draw [draw=red]
([xshift=-27pt,yshift=-15pt]safeC.south)

rectangle ([xshift=75pt,yshift=70pt]iunsafeC.north);
\end{pgfonlayer}
\end{tikzpicture}

};
\end{minipage}
\vspace{2mm}
\caption{Precondition inference: Reality (left), Initial approximations (middle), One step refinement   of approximations using Algorithm \ref{alg:precond} (right). We discard the regions classified
 as definitely safe or unsafe considering only the intersection.}
\label{fig:alg-progress}
\end{figure}

The progress of Algorithm \ref{alg:precond} is depicted in Figure \ref{fig:alg-progress}. The leftmost panel reflects the reality, that is, it depicts the set of concrete safe and unsafe states along with the set of initial states leading to them. The set of safe states and the set of initial states leading to them are respectively marked as $\mathtt{Safe}$ and $\mathtt{Safe^I}$ and are drawn on the left side of each panel. Similarly, the set of unsafe states and the set of initial states leading to them are respectively marked as $\mathtt{Unsafe}$ and $\mathtt{Unsafe^I}$ and are shown on the right side of each panel.  Given a program together with the description of the sets of interest (safe and unsafe states), the precondition analysis  infers a set of initial states that leads to these sets of interest. Due to approximations, it discovers  over-approximations of the initial sets of states. These over-approximations are depicted as ellipses in the middle panel with the corresponding colours for safe and unsafe initial states.  These approximations may intersect. As a consequence, there are some witness traces  from the left ellipse  to $\mathtt{Unsafe}$ (shown by dotted arrows) and vice versa.  The algorithm aims to reduce these ellipses progressively  towards the optimal condition, where the left ellipse leads to the set of safe states and the right one to the set of unsafe states (that is, these ellipses do not overlap anymore). A single step of refinement is illustrated in the rightmost panel.

In the algorithm, the following operations are carried out in an iterative manner and possibly in parallel (within the \emph{while} loop). The instructions on two sides of the boxes can be executed in parallel. One or more of the transformation of $P_s$ and $P_u$  with respect to \safe~and \unsafe~respectively are carried out  and the NPs are extracted from the resulting programs (\emph{line 6}). The algorithm terminates and returns an SP if the conjunction of these NPs is unsatisfiable (\emph{line 10}, \emph{optimal}) or it is not stronger (in the sense of implication of formulas) than $\phi_{old}$ (\emph{line 13}). Otherwise, the algorithm  iterates with the new under-approximations obtained by restricting their initial states with the conjunction (\emph{line 16}). Note that the  conjunction needs to be converted to DNF before applying Definition \ref{def:ris}, which may blow up the number of resulting clauses  (only the initial ones).
\medskip

\noindent
\textbf{Simulating the algorithm on the example program}. We denote the program  in Figure \ref{ex:precond} by $P$.  Initially,  $\phi_{old}=\trueit$, the initial state of $P$. First, we choose to apply PE   to $P$  with respect to \safe~ and \unsafe~ respectively (call it the iteration 0 of the algorithm). The resulting set of CHCs are shown in Figure \ref{pe_example} along with the derived NPs. Since the conjunction of  $\necpre{s}$ and $\necpre{u}$, that is, $\mathtt{B\geq 0 \wedge A \ge 1}$, is  \emph{satisfiable}, the  preconditions are not optimal. But the conjunction  is  stronger than $\phi_{old}$ (conjunction of previous approximations), so the algorithm progresses to refinement (\emph{line 14-16}).  Before under-approximating the programs, we update the SPs from this iteration  as:
\begin{flalign*}
\suffpre{s} & \equiv \mathtt{\falseit \vee (B\geq 0 \wedge \neg (B<0 \vee A\geq 1)) \equiv B\ge 0 \wedge A\leq 0.} \\
\suffpre{u} & \equiv \mathtt{\falseit \vee ((B<0 \vee A\geq 1) \wedge \neg B\geq 0) \equiv B<0.}
\end{flalign*}
The results of under-approximating the programs  in Figure \ref{pe_example} by $\mathtt{B\geq 0 \wedge A \ge 1}$ following Definition \ref{def:ris}  are trivial and are not shown here. In the next iteration (iteration 1), we apply CS  with respect to \safe~ and \unsafe~ respectively to the under-approximations, obtaining the clauses shown in Figure \ref{cs_after_ris}.
\begin{figure}[t]
\centerline{
  \begin{tabular}{ll}
    \pcode[\small]{
$\mathtt{unsafe \leftarrow  A<0, B=0, wh\_2(B,A) }$. \\ 
$\mathtt{wh\_2(A,B) \leftarrow B<0,  A=0, C=1,}$ \\
~~~~~~~~~~~~~~~~$\mathtt{B-D= -1,  wh\_1(C,D)}$. \\ 
$\mathtt{wh\_1(A,B) \leftarrow  A>B,  B\geq 0,  A\geq 1, init(A,B)}$. \\ 
$\mathtt{wh\_1(A,B) \leftarrow  A>B,  A\geq 1,  A-C= -1, }$ \\
~~~~~~~~~~~~~~~~$\mathtt{B-D= -1,  wh\_1(C,D)}$. \\ 
$\mathtt{init(A,B)\leftarrow A>B,  B\geq 0}$.  
    } &
    \pcode[\small]{
$\mathtt{safe \leftarrow A\geq0, B=0, wh\_2(B,A)}$. \\ 
$\mathtt{wh\_2(A,B) \leftarrow B\geq0,  A=0,  C=1,  }$ \\ 
~~~~~~~~~~~~~~~~$\mathtt{B-D= -1,wh\_1(C,D)}$. \\ 
$\mathtt{wh\_1(A,B) \leftarrow  B\geq A,   A\geq 1, init(A,B)}$. \\ 
$\mathtt{wh\_1(A,B) \leftarrow  B \geq A,   A-C= -1, }$ \\  
~~~~~~~~~$\mathtt{A \geq 1,B-D= -1,  wh\_1(C,D)}$. \\ 
$\mathtt{init(A,B) \leftarrow B \geq A, A \geq 1}$.
} 
  \end{tabular}
}
\caption{Constraint specialised programs: wrt \texttt{unsafe} (left) and wrt \texttt{safe} (right)\label{cs_after_ris}}
\end{figure}
Then the NPs are: 
%\begin{flalign*}
$\necpre{s} \equiv \mathtt{(B\ge A \wedge A\geq 1)}$ and 
$\necpre{u} \equiv \mathtt{(B \ge 0 \wedge A > B).}$
%\end{flalign*}
Since $\necpre{s} \wedge \necpre{u} \equiv \falseit$, the preconditions are optimal and the algorithm terminates. The final SPs are given as the disjunction of SPs over the iterations as follows: $\suffpre{s} \equiv \mathtt{(B\ge 0 \wedge A \leq 0) \vee (B\ge A \wedge A\geq 1)}$ and $\suffpre{u} \equiv \mathtt{B<0 \vee (B \ge 0 \wedge A > B)}.$
%\begin{flalign*}
%\mathtt{\suffpre{s} \equiv (B\ge 0 \wedge A \leq 0) \vee (B\ge A \wedge A\geq 1).}  ~~~~
%\mathtt{\suffpre{u} \equiv B<0 \vee (B \ge 0 \wedge A > B).}
%\end{flalign*}

If instead we apply the sequence $CS$, $PE$ to the original program, we get  
%\begin{flalign*}
$\necpre{s} \equiv \mathtt{(B\ge0 \wedge A\leq 0) \vee (A\geq 1 \wedge B\geq A)}$,
$\necpre{u} \equiv \mathtt{(B<0 \wedge A\leq 0) \vee (A\geq 1 \wedge A> B)}$ as \emph{optimal} preconditions. 
%\end{flalign*}
This suggests that a sequence of transformations can potentially avoid the costly DNF conversion  (needed during refinement)  after each transformation.
Therefore, in our experiments, we  apply the sequence of transformations PE followed by  CS followed by TE, rather than a single one. The rationale behind this particular sequences is that CS is most effective when performed  after PE, which propagates constraints and introduces new versions of predicates. Trace elimination helps decompose the problem as well as splitting predicates, which induces a split in the resulting programs. 

In each iteration of the loop, the algorithm either explores a strictly smaller set of initial states that have not yet been classified as safe or unsafe or, terminates. This property is formalised in Proposition \ref{prop:progress}.

\begin{proposition}[Progress and Termination of Algorithm \ref{alg:precond}]
\label{prop:progress}
At each iteration of the \emph{while} loop, Algorithm \ref{alg:precond} either progresses (considers a smaller set of states that is yet to be proven safe or unsafe)  or terminates.
\end{proposition}
\begin{proof}[sketch]
At each iteration, the algorithm maintains a constraint formula (representing the initial set of states that is yet to be classified as safe or unsafe). The algorithm can exit if the formula in the current iteration is the formula $\falseit$ (\emph{line 8-10}), or is equivalent to the formula in the previous iteration (\emph{line 11-13}). In other case, the current formula is stronger (representing a smaller set of initial states, treating constraint formula as a set) than that in the previous iteration, so it recurses (\emph{line 14-16}).  This means the set that is yet to be proven safe or unsafe gets smaller. In this sense, the algorithm makes progress. 
\end{proof}

% The partial sufficient preconditions are updated with the set of states that have already been proven safe and unsafe (\emph{line 16}) before   under-approximations  (\emph{line 16}). The under-approximations carry the states that are still to be classified to the next iteration. 
 Proposition \ref{prop:sequence} establishes the correctness  of the combination of transformations, Proposition \ref{prop:lift_precond} ensures that the precondition of an under-approximation of CHCs is also that of the original and Proposition \ref{prop:disj_precond} allows us to combine preconditions derived in each  iterations.All these propositions together ensure the soundness of  Algorithm \ref{alg:precond}, which is formally stated as follows. 

\begin{theorem}[Soundness of Algorithm \ref{alg:precond}]
\label{prop:soundness}
Let $P$ be a set of CHCs  annotated with the predicates \safe~(indicating a set of safe terminating states) and \unsafe~(indicating a set of unsafe states). If Algorithm \ref{alg:precond} returns a tuple $\tuple{S,U}$, then $S$ and $U$ are the  SPs for safety and unsafety of $P$, respectively, with respect to the predicates \safe~and \unsafe.
\end{theorem}
\begin{proof}[sketch]
In each iteration the algorithm essentially computes over approximations of the set of \emph{safe} and \emph{unsafe} states using the sequence of transformations. Proposition \ref{prop:sequence} establishes the correctness  of the sequence. Then it computes the
under-approximations of the original program with the sets of states that are yet to be classified as \emph{safe} or \emph{unsafe} while accumulating  those states (preconditions) that are known to be \emph{safe} and known to be \emph{unsafe}. The under-approximations are fed to the next iteration. Proposition \ref{prop:lift_precond} ensures that the precondition of an under-approximation of CHCs is also that of the original.  Proposition \ref{prop:disj_precond} ensures that the preconditions accumulated (disjunctively) in each iteration are the preconditions for the original program. These arguments together ensures the soundness of the algorithm. \qed
\end{proof}
\paragraph{Limitations of our method and non-termination.} 
The SPs derived by our method may include non-terminating inputs, 
that neither lead to safe nor unsafe.
Popeea and Chin~\cite{PopeeaC13-dualanalysis} treat such inputs as 
\emph{unsafe} whereas Seghir and Schrammel~\cite{SeghirS14-precond} 
ignore them, as do we. 
However, the modelling of safe and unsafe terminating states and their 
over-approximations allow us to reason about a limited form of 
non-termination as suggested by Popeea and Chin~\cite{PopeeaC13-dualanalysis}. 
That is, any input state that is neither in the over-approximation of 
safe nor unsafe is non-terminating 
(again assuming that we model all the terminating safe and unsafe states). 

\begin{wrapfigure}[10]{r}{4.4cm}
\centering{
\vspace*{-4ex}
    \pcode[\small]{
   \textbf{void} main(\textbf{int}~ a) \{ \\
   ~~\textbf{while} ($a \geq 0 $)  \{\\
   ~~~~\textbf{if} ($a \leq 9$) 
   ~~{a++;} \\
   ~~~~\textbf{else if} ($a==10$) \\
   ~~~~~~{a = 5;} \\
   ~~~~\textbf{else}~\textbf{return}; \\
   ~~\} \\
   ~~\textbf{assert} ($\id{false}$); \\
   \}          
      }    
}
  \caption{Non-termination}
 \label{ex:non-termination}
\end{wrapfigure}
We illustrate this with the example in Figure \ref{ex:non-termination}. This program does not terminate if $a \in [ 0,10 ]$. We derive $\necpre{s}= \mathtt{a} \geq 11~ \text{and} ~\necpre{u}= \mathtt{a} < 0$ as NPs. 
Thus the condition satisfying $\neg (\necpre{s} \vee \necpre{u})$, that is, $\mathtt{\neg(a \geq 11 \vee a <0)} \equiv \mathtt{a \in [ 0,10 ]}$ is a precondition for non-termination (which happens to be precise in this case). It is obtained as our method's side-effect and we leave the primary analysis of  non-termination for future work.

\section{Experimental evaluation}
\label{sec:experiments}
%\bk{In Table 1, please state that the non-trivial cases exclude the optimal cases.
%- Please state that each row contains the instances for which the algorithm
%terminated after the given number of iterations.
%- It would be nice to see the results per benchmark category.
%- You could consider to compare the results to verification results. E.g., when
%you derive a non-trivial sufficient condition for unsafety then you know that
%the program is unsafe.}
%\bk{-how these benchmarks were identified 
%- benchmarks set are not superset, why benchmarks are removed. when I look at
%the results in
%[23] the approach listed in this submission as WP-Rahft gives results for all
%iterations 0,1,2,3 in [23] with the same timeout of 300 seconds, whereas in the
%current submission there are no results for 2 and 3. The reason is probably that
%the benchmark set for this paper is not a superset of the benchmark set for the
%previous paper. I understand if the set gets larger, but why were benchmarks
%removed? 
%}

We implemented Algorithm \ref{alg:precond} (a sequential version) in the tool \pihorn~(Precondition Inferrer for Horn clauses), which is available from \url{https://github.com/bishoksan/PI-Horn}. Given 3 different program transformations, the implementation fixes the length of the sequence to 3. The order of their application is PE followed by CS followed by TE. The tool is written in Ciao Prolog~\cite{Ciao} and makes use of the Parma Polyhedra Library (PPL)~\cite{BagnaraHZ08SCP} and Yices2 SMT solver~\cite{Dutertre:cav2014}. 
%The tool applies a seqence of program transformations in an iterative manner (also called refinements) until some termination criteria are met or resources are exhausted. 
The input to our tool is a set of CHCs;  with \safe, \unsafe~and \init~as distinguished predicates (for technical reasons). It outputs a pair of SPs for safety and unsafety. The derived SPs are classified as: (i) \emph{optimal}: the precondition is both necessary and sufficient (exact); (ii) \emph{(un)safe-non-trivial}: the SP for (un)safety is different from $\falseit$ (but not \emph{optimal}) and (iii) \emph{(un)safe-trivial}: the SP for (un)safety  is  $\falseit$. 
%\begin{itemize}
%\item \emph{optimal}: the precondition is both necessary and sufficient (exact).
%\item \emph{(un)safe-non-trivial}: the sufficient (un)safe precondition is different from $\falseit$ (excluding \emph{optimal}).
%%\item \emph{unsafe-non-trivial}: the sufficient unsafe precondition is different from $\falseit$ (excluding \emph{optimal}).
%\item \emph{(un)safe-trivial or unsafe-trivial}: the sufficient (un)safe precondition  is  $\falseit$.
%\end{itemize}
\medskip

\noindent
\textbf{Benchmarks and experimental settings.} We tested our approach with 261 integer programs (available from \url{https://github.com/bishoksan/PI-Horn/tree/master/benchmarks}) 
%of up to approximately 500 lines of code 
collected from the following sources which we think could be the good candidates.   
(i) 150 programs from the loop (69) and recursive (81) categories of SV-COMP'18~\cite{svcomp17}; 
(ii) 83 programs from the DAGGER~\cite{DBLP:conf/tacas/GulavaniCNR08} and TRACER tools~\cite{DBLP:conf/cav/JaffarMNS12} and 
(iii) 28 programs from the literature on precondition inference and backwards analysis~\cite{DBLP:conf/sas/BakhirkinBP14,DBLP:journals/entcs/Mine12,DBLP:conf/vmcai/Moy08,DBLP:conf/sas/BakhirkinM17,DBLP:conf/rp/CassezJL17}. Unfortunately, we could not include all the benchmarks used by Kafle et al.~\cite{kafle-iclp18} since some of their C sources were not available, but only their CHCs representations modelling only the \emph{unsafe} terminating states.  
Benchmark set (i) is designed for software verification competitions  while (ii) and (iii) are designed to demonstrate the strength of some specific tools and techniques. We adapt these programs written in C for precondition inference in the following way. They
 are first translated to CHCs by specialising an interpreter of C written in Prolog using the tool VeriMap~\cite{verimap}.  The specialisation approach was described by De Angelis et al.~\cite{DBLP:journals/scp/AngelisFPP17}, and models only the states leading to an ``error'' state by the predicate \unsafe.  The current translation also models  \emph{return} statements and the initial states of programs  by two special predicates \safe~and \init~respectively (again using VeriMap), which allows us to reason about the safe terminating states and the initial states respectively. Given a sequence of initial state variables $x_1,  \ldots, x_n$ and the corresponding initial clause $\init(x_1,  \ldots, x_n) \leftarrow \phi(x_1,  \ldots, x_n)$ obtained by the translation, we replace $\phi(x_1, x_2, \ldots, x_n)$ by $\trueit$ obtaining the clause $\init(x_1, x_2, \ldots, x_n) \leftarrow \trueit$. The replacement allows the analysis to infer preconditions in terms of these variables starting from an unrestricted set of \emph{initial} states. 
% See for instance the example program in Figure \ref{ex:precond} where we replace $\mathtt{init(A,B) \leftarrow A=0, B=0}$ by $\mathtt{init(A,B) \leftarrow \trueit}$. In other words, we treat the \emph{main} method with arguments $A$ and $B$.
 
% The recent works on precondition inference are of Seghir et al. \cite{SeghirK13-precond,SeghirS14-precond}, Bakhirkin et al.\ \cite{DBLP:conf/sas/BakhirkinBP14} and  of Kafle et al.\ \cite{kafle-iclp18}. 
Unfortunately, we could not compare our tool against the recent works of Seghir et al.~\cite{SeghirK13-precond,SeghirS14-precond} and that of Bakhirkin et al.~\cite{DBLP:conf/sas/BakhirkinBP14}, the first due to some issues with the tool (discovered together with the authors), and the second due to the lack of a tool automating their approach (confirmed by the authors via email).
%Together with the author(s) of \cite{SeghirK13-precond,SeghirS14-precond} we found some issues with their tool making us unable to compare; whereas Bakhirkin confirmed via email that they do not yet have a tool that automates their approach.  
But we do compare to the work of Kafle et al.~\cite{kafle-iclp18} (represented in the table as \wprahft). Although some of the components used in our tools are the same,  a direct comparison is difficult  due to different outcomes. That tool, for example, cannot detect optimality and does not attempt to do so, unlike ours. Therefore, it is hard to compare the quality of the generated preconditions. 

Experiments have been carried on a MacBook Pro, 
running OS X 10.11 in 16GB Memory and 2.7 GHz Intel Core i5 processor. 
\medskip
 
\noindent
\textbf{Experimental Results.} 
Table \ref{tbl:exp-results} presents the results. 
The columns respectively show 
\textsf{iter} (\# of refinement iterations (\emph{line 14-16})), 
\textsf{opt} (\# of programs with optimal preconditions), 
\textsf{ntS} (\textsf{Sw}) (\# of programs with non-trivial SPs for safety excluding optimal cases, and, in parentheses, the difference with trivial SPs for unsafety),  
\textsf{ntU} (\textsf{Uw}) (same, but for unsafety), 
\textsf{ntSU} (\# of programs with non-trivial SPs for both safety and unsafety),  
\textsf{tSU} (\# of programs with trivial SPs for both safety and unsafety), 
\textsf{total}/\textsf{iter} (\# of programs with non-trivial plus optimal SP per iteration), 
\wprahft~\textsf{total}/\textsf{iter} (\# of programs with non-trivial SP 
per iteration for \wprahft). 
% The numbers in parentheses (in the $3^{rd}$ and $4^{th}$ columns) represent 
% the difference in the number of programs with non-trivial SPs for safety 
% and unsafety. 
For example, the entry 9 (7) in column 3 indicates that there were 9
non-trivial SPs for safety, of which 7 had trivial SPs for unsafety. Each row contains the number of instances for which the algorithm terminated after the given number of iterations. 
\medskip

\begin{table}[t]
\newcommand{\ph}{\hphantom{1}}
\centering
\resizebox{\textwidth}{!}{  
    \begin{tabular}{|r|r|r|r|r|r|r|r|r|}
    \hline
   \centering{\textsf{iter}} & \textsf{opt} & \multicolumn{1}{|p{1.5cm}|}{\centering \textsf{ntS} (\textsf{Sw})} & \multicolumn{1}{|p{1.5cm}|}{\centering \textsf{ntU} (\textsf{Uw})} & \multicolumn{1}{|p{1.5cm}|}{\centering \textsf{ntSU}} & \multicolumn{1}{|p{1.5cm}|}{\centering \textsf{tSU} }& &  \textsf{total}/\textsf{iter} & \multicolumn{1}{|p{2cm}|}{\centering \wprahft  \\ \textsf{total}/\textsf{iter}}\\ \hline \hline
    0       & 58    & 0 \ph(0)         & 0 (0)         & 0                  & 0               && 58     &    197  \\ \hline
    1       & 87    & 9 \ph(7)         & 5 (3)         & 2                  & 20             && 99      &   20 \\ \hline
    2       & 21    & 20 (15)       & 7 (2)         & 5                  & 0              && 43    &    0   \\ \hline
    3       & 5     & 6 \ph(3)         & 3 (0)         & 3                  & 0              && 11       &  0  \\ \hline
    4       & 2     & 3 \ph(1)         & 2 (0)         & 2                  & 0              && 5        &   0 \\ \hline
    5       & 2     & 0 \ph(0)         & 0 (0)         & 0                  & 0              && 2       &   1  \\ \hline
    6       & 1     & 0 \ph(0)         & 0 (0)         & 0                  & 0              && 1        &   0 \\ \hline \hline
    \#total & 176   & 38 (26)       & 17 (5)        & 12                 & 20             && 219  & 218        \\ \hline
    \end{tabular}
    }
    \vspace{2mm}
    \caption{Experimental results  on  261 programs, with a timeout of 300 seconds.}
    \label{tbl:exp-results}
    \vspace{-3em}
\end{table}

 %\end{wrapfigure}
The results show our tool is able to generate non-trivial preconditions 
for $83\%$ of the programs while timing out for $9\%$.
The average time taken (ignoring timeouts) per each instance is 5.94 seconds. 
Optimal preconditions are inferred for $67\%$ of the programs whereas for 
$8\%$ the tool fails to infer any meaningful preconditions. 
A meaningful SP for safety or unsafety ($\neq \falseit$) apart from 
optimal ones is inferred for $16\%$ of the problems 
(columns $3^{rd}$ and $4^{th}$).   
Preprocessing using partial evaluation and constraint specialisation, 
as well as refinement by removing program traces and restarting the 
analysis from unresolved initial states, increases the precision of the 
analysis for the following reasons. 
On the one hand, partial evaluation brings polyvariant specialisation by 
creating several versions of the predicates, which the abstract interpreter 
in constraint specialisation can take advantage of. 
On the other hand, the refinement helps limit the set of initial 
states to explore and/or reduces the number of paths to explore. 
Note that the tool infers optimal preconditions for 58 programs due 
to preprocessing alone, which is indicated by the result along the row 
\emph{iteration} 0 in the table, whereas it infers optimal preconditions 
for 99 programs after the first refinement. 
However, it should be noted that the refinement also hinders the 
performance of the analysis by increasing the size of the resulting programs. 
This explains why we solve fewer and fewer problems after iteration 2. 
The number of programs with non-trivial SPs are almost the same for both 
the tools (219 vs. 218), but they differ in the quality (such as optimality) 
of the generated SPs. 
The refinement in our approach allows us to derive optimal or non-trivial 
preconditions for more programs, but in case of \wprahft~it does not yield 
any better results after the first refinement iteration 
(note though that the derived preconditions may be more general).  
The average time taken (ignoring the timeouts
\pds{as Graeme pointed out, this is the wrong way to handle timeouts when computing averages}) 
by \wprahft\ per each instance is 5.59 (first \emph{iteration}) 
and 8.34 (sixth \emph{iteration}) seconds. 
The first one is bit lower and the second is bit higher than in our approach. 
The reason for this difference is that \wprahft\ does not know when to stop 
refining since the termination criterion, the maximum number of iterations, 
is rather weak. 
Another possible reason is that at each iteration an under-approximation 
of the original program is analysed by our tool as opposed to \wprahft, 
which can be easier. Though the refinement may help derive a 
weaker SP, it also consumes more time.

\section{Related work}\label{sec:rel}

%Given a wide range of applications of preconditions, 
Given a wide range of applications, there are a handful of techniques in the literature dealing with precondition inference. We classify them as follows.
\medskip

\noindent
\textbf{Abstract interpretation:}
%Forward abstract interpretation  approximates the set of reachable states of a program while backward abstract interpretation approximates the set of states that can reach some goal states. These can produce over-or under-approximations depending on the transfer functions used. However, the forward over-approximating abstract interpretation is more common and better developed than the under-approximating backward abstract interpretation \cite{DBLP:conf/sas/BakhirkinBP14}. But combining these analyses  offers added advantages \cite{CousotCousot92-3,CCL_VMCAI11,DBLP:conf/sas/BakhirkinM17}. 
Over-approximation techniques (forward and backwards abstract interpretations or
their combination~\cite{CousotCousot92-3,CCL_VMCAI11,DBLP:conf/sas/BakhirkinM17}) inherently derive NPs and complementation allows deriving SPs at the cost of precision (due to approximation of the complement). In order to minimise the loss, we  need to find either a (pseudo) complemented domain~\cite{Marriott-Sondergaard-LOPLAS93} or   an alternative to complementation that is less aggressive.  Along these lines, Howe et al.~\cite{DBLP:conf/lopstr/HoweKL04} utilise the pseudo-complemented domain \emph{Pos}  to infer SPs whereas Bakhirkin et al.~\cite{DBLP:conf/sas/BakhirkinBP14} exchange an abstract complement operation for  abstract logical subtraction. In contrast to these lines of work, we neither assume an abstract domain to be complemented nor apply complementation of an abstract element during abstract interpretation. Instead our method applies to any abstract domain, and the complementation is carried out as a separate process outside the abstract domain, storing the result as a formula without any loss of precision.

Little work has been done that inherently computes SPs without complementation.
The notable exception is the work of Min\'e~\cite{DBLP:journals/entcs/Mine12}, which designs all required purpose-built backward transfer functions for intervals, octagons and convex polyhedra domains. The downside is that  the purpose-built operations, including widening, can be rather intricate and require substantial implementation effort. Moy~\cite{DBLP:conf/vmcai/Moy08} employs weakest-precondition reasoning and forward abstract interpretation to attempt to generalise conditions at loop heads to infer SPs. The derived conditions offer limited use except  for a theorem prover. As opposed to these works, our method makes use of standard tools and techniques offering simplicity and ease of implementation. Furthermore, our results can be consumed by other tools in the area of program analysis and verification.
Though formulated in the context of software verification, the dual-analysis approach of Poppea and Chin~\cite{PopeeaC13-dualanalysis} uses over-approximations as in our approach to concurrently infer NPs for safety and unsafety of each method (and program), from which SPs can be derived.  However, in contrast to our approach, no refinement of the derived preconditions takes place.
\medskip

\noindent
\textbf{Program transformation:}
The forward and backward iterative specialisation approach of De Angelis et al.~\cite{DBLP:journals/scp/AngelisFPP14} designed to verify program properties can be used for deriving preconditions, as in our approach. The transformation approach uses constraint generalisation method instead of abstract interpretation and can be adjusted to refine  preconditions.
\medskip

\noindent
\textbf{Counterexample-guided abstraction refinement (CEGAR):}
Seghir et al.~\cite{SeghirK13-precond,SeghirS14-precond} use a CEGAR approach to derive \emph{exact} necessary and sufficient preconditions for safety. Like us, they model safe and unsafe states of a program and refine their approximations until they are disjoint. Their algorithm may diverge due to (i)  the lack of a suitable generalisation of the counterexamples (an inherent limitation of CEGAR) and (ii) the termination condition (disjointness) that is too hard to achieve for realistic programs (due to undecidability). We, on the other hand, use abstract interpretation and program transformation, so each step of the algorithm terminates and a sound precondition can be derived from the resulting programs. In addition, optimality is not the end goal for us and it is obtained as a by-product of precision refinement.  

The work of Kafle et al.\ \cite{kafle-iclp18} is orthogonal to all of the above; it combines a range of established tools and techniques such as abstract interpretation, CEGAR and program transformations in a profitable way. The iterative nature of their approach allows them to derive more precise preconditions for safety on demand, however the termination criterion, the maximum number of iterations supplied by the user, is rather weak. The quality of the preconditions such as optimality cannot be checked using their approach. The current work offers several advantages over that approach. We model both safe and unsafe states of programs. This allows detecting optimality and also inferring NP and SP for both safety and unsafety. In addition, it allows reasoning about a limited form of non-termination and provides more refined termination criteria. Unlike many methods in the literature~\cite{SeghirK13-precond,SeghirS14-precond,DBLP:conf/sas/BakhirkinBP14}, our method can handle not only programs with procedures, but also recursive programs, in a uniform way.

\section{Concluding remarks}
\label{sec:conclusion}

We have presented an iterative method for automatically deriving sufficient  preconditions for safety and unsafety of  programs. It maintains  over-approximations of the set of \emph{safe} and \emph{unsafe} initial states. At each iteration, only the set of states that are common to these over-approximations are considered, as those are yet to be classified as safe or unsafe. The method terminates when the common set of states is empty or the common set does not get smaller in successive iterations. Experimental results show that the method is able to generate optimal preconditions in $67\%$ of the cases and also  solves some problems which are otherwise unsolvable  using only approximations of the unsafe states (as was done in previous work). 
Owing to over-approximation, the sufficient preconditions may include some non-terminating states, which hinders the derivation of optimal preconditions. 
In future work we intend to augment our method with non-termination analysis.

\paragraph{\textbf{Acknowledgement}}
We thank John Gallagher for comments on an earlier draft and many useful 
discussions which has led to many improvements. 
We thank Emanuele de Angelis for C to CHC translation, 
help with the tool VeriMAP, benchmarks and many stimulating discussions. 

\bibliographystyle{abbrv}
\bibliography{refs}

\begin{thebibliography}{10}

\bibitem{BagnaraHZ08SCP}
R.~Bagnara, P.~M. Hill, and E.~Zaffanella.
\newblock The {Parma Polyhedra Library}: Toward a complete set of numerical
  abstractions for the analysis and verification of hardware and software
  systems.
\newblock {\em Sci.\ Comput.\ Program.}, 72(1--2):3--21, 2008.

\bibitem{DBLP:conf/sas/BakhirkinBP14}
A.~Bakhirkin, J.~Berdine, and N.~Piterman.
\newblock Backward analysis via over-approximate abstraction and
  under-approximate subtraction.
\newblock In {\em Static Analysis}, volume 8723 of {\em LNCS}, pages 34--50.
  Springer, 2014.

\bibitem{DBLP:conf/sas/BakhirkinM17}
A.~Bakhirkin and D.~Monniaux.
\newblock Combining forward and backward abstract interpretation of {Horn}
  clauses.
\newblock In {\em Static Analysis}, volume 10422 of {\em LNCS}, pages 23--45.
  Springer, 2017.

\bibitem{svcomp17}
D.~Beyer.
\newblock Software verification with validation of results (report on {SV-COMP}
  2017).
\newblock In {\em Proc.\ TACAS 2017}, volume 10206 of {\em LNCS}, pages
  331--349, 2017.

\bibitem{DBLP:conf/rp/CassezJL17}
F.~Cassez, P.~G. Jensen, and K.~G. Larsen.
\newblock Refinement of trace abstraction for real-time programs.
\newblock In {\em Reachability Problems}, volume 10506 of {\em LNCS}, pages
  42--58. Springer, 2017.

\bibitem{CousotCousot92-3}
P.~Cousot and R.~Cousot.
\newblock Abstract interpretation and application to logic programs.
\newblock {\em J. Log. Program.}, 13(2{\&}3):103--179, 1992.

\bibitem{CCL_VMCAI11}
P.~Cousot, R.~Cousot, and F.~Logozzo.
\newblock Precondition inference from intermittent assertions and applications
  to contracts on collections.
\newblock In {\em VMCAI 2011}, volume 6538 of {\em LNCS}, pages 150--168.
  Springer, 2011.

\bibitem{DBLP:journals/scp/AngelisFPP14}
E.~{De Angelis}, F.~Fioravanti, A.~Pettorossi, and M.~Proietti.
\newblock Program verification via iterated specialization.
\newblock {\em Sci.\ Comput.\ Program.}, 95:149--175, 2014.

\bibitem{verimap}
E.~De~Angelis, F.~Fioravanti, A.~Pettorossi, and M.~Proietti.
\newblock {VeriMAP}: A tool for verifying programs through transformations.
\newblock In {\em TACAS 2014}, volume 8413 of {\em LNCS}, pages 568--574.
  Springer, 2014.

\bibitem{DBLP:journals/scp/AngelisFPP17}
E.~{De Angelis}, F.~Fioravanti, A.~Pettorossi, and M.~Proietti.
\newblock Semantics-based generation of verification conditions via program
  specialization.
\newblock {\em Sci.\ Comput.\ Program.}, 147:78--108, 2017.

\bibitem{Dutertre:cav2014}
B.~Dutertre.
\newblock Yices 2.2.
\newblock In {\em Proc.\ CAV 2014}, volume 8559 of {\em LNCS}, pages 737--744.
  Springer, 2014.

\bibitem{gallagher:pepm93}
J.~P. Gallagher.
\newblock Specialisation of logic programs: A tutorial.
\newblock In {\em PEPM 1993}, pages 88--98, 1993.

\bibitem{Gallagher-Lafave-Dagstuhl}
J.~P. Gallagher and L.~Lafave.
\newblock Regular approximation of computation paths in logic and functional
  languages.
\newblock In {\em Partial Evaluation}, volume 1110 of {\em LNCS}, pages
  115--136. Springer, 1996.

\bibitem{DBLP:conf/pldi/GrebenshchikovLPR12}
S.~Grebenshchikov, N.~P. Lopes, C.~Popeea, and A.~Rybalchenko.
\newblock Synthesizing software verifiers from proof rules.
\newblock In {\em PLDI 2012}, pages 405--416. {ACM}, 2012.

\bibitem{DBLP:conf/tacas/GulavaniCNR08}
B.~S. Gulavani, S.~Chakraborty, A.~V. Nori, and S.~K. Rajamani.
\newblock Automatically refining abstract interpretations.
\newblock In {\em TACAS 2008}, volume 4963 of {\em LNCS}, pages 443--458.
  Springer, 2008.

\bibitem{DBLP:conf/cav/GurfinkelKKN15}
A.~Gurfinkel, T.~Kahsai, A.~Komuravelli, and J.~A. Navas.
\newblock The {SeaHorn} verification framework.
\newblock In {\em CAV 2015}, volume 9206 of {\em LNCS}, pages 343--361.
  Springer, 2015.

\bibitem{Ciao}
M.~V. Hermenegildo, F.~Bueno, M.~Carro, et~al.
\newblock An overview of {Ciao} and its design philosophy.
\newblock {\em Theory and Practice of Logic Programming}, 12(1-2):219--252,
  2012.

\bibitem{DBLP:conf/lopstr/HoweKL04}
J.~M. Howe, A.~King, and L.~Lu.
\newblock Analysing logic programs by reasoning backwards.
\newblock In {\em Program Development in Computational Logic}, volume 3049 of
  {\em LNCS}, pages 152--188. Springer, 2004.

\bibitem{DBLP:conf/cav/JaffarMNS12}
J.~Jaffar, V.~Murali, J.~A. Navas, and A.~E. Santosa.
\newblock {TRACER}: A symbolic execution tool for verification.
\newblock In {\em CAV 2012}, volume 7358 of {\em LNCS}, pages 758--766.
  Springer, 2012.

\bibitem{Jones-Gomard-Sestoft}
N.~Jones, C.~Gomard, and P.~Sestoft.
\newblock {\em {P}artial {E}valuation and {A}utomatic {S}oftware {G}eneration}.
\newblock Prentice Hall, 1993.

\bibitem{DBLP:journals/scp/KafleG17}
B.~Kafle and J.~P. Gallagher.
\newblock Constraint specialisation in {Horn} clause verification.
\newblock {\em Sci. Comput. Program.}, 137:125--140, 2017.

\bibitem{DBLP:journals/cl/KafleG17}
B.~Kafle and J.~P. Gallagher.
\newblock Horn clause verification with convex polyhedral abstraction and tree
  automata-based refinement.
\newblock {\em Computer Languages, Systems {\&} Structures}, 47:2--18, 2017.

\bibitem{kafle-iclp18}
B.~Kafle, J.~P. Gallagher, G.~Gange, et~al.
\newblock An iterative approach to precondition inference using constrained
  {Horn} clauses.
\newblock {\em Theory and Practice of Logic Programming}, 18:553--570, 2018.

\bibitem{Marriott-Sondergaard-LOPLAS93}
K.~Marriott and H.~S{\o}ndergaard.
\newblock Precise and efficient groundness analysis for logic programs.
\newblock {\em ACM LOPLAS}, 2(1--4):181--196, 1993.

\bibitem{DBLP:journals/entcs/Mine12}
A.~Min{\'{e}}.
\newblock Inferring sufficient conditions with backward polyhedral
  under-approximations.
\newblock {\em Electr. Notes Theor. Comput. Sci.}, 287:89--100, 2012.

\bibitem{DBLP:conf/vmcai/Moy08}
Y.~Moy.
\newblock Sufficient preconditions for modular assertion checking.
\newblock In {\em VMCAI 2008}, volume 4905 of {\em LNCS}, pages 188--202.
  Springer, 2008.

\bibitem{Peralta-Gallagher-Saglam-SAS98}
J.~C. Peralta, J.~P. Gallagher, and H.~Sa\u{g}lam.
\newblock Analysis of imperative programs through analysis of constraint logic
  programs.
\newblock In {\em Static Analysis}, volume 1503 of {\em LNCS}, pages 246--261,
  1998.

\bibitem{PopeeaC13-dualanalysis}
C.~Popeea and W.~Chin.
\newblock Dual analysis for proving safety and finding bugs.
\newblock {\em Sci. Comput. Program.}, 78(4):390--411, 2013.

\bibitem{SeghirK13-precond}
M.~N. Seghir and D.~Kroening.
\newblock Counterexample-guided precondition inference.
\newblock In {\em {ESOP} 2013}, volume 7792 of {\em LNCS}, pages 451--471.
  Springer, 2013.

\bibitem{SeghirS14-precond}
M.~N. Seghir and P.~Schrammel.
\newblock Necessary and sufficient preconditions via eager abstraction.
\newblock In {\em APLAS 2014}, volume 8858 of {\em LNCS}, pages 236--254.
  Springer, 2014.

\end{thebibliography}
\end{document}